\documentclass{article}
\usepackage{amsmath,amsfonts,amssymb,amsthm}

\usepackage{graphicx}

\allowdisplaybreaks
\newtheorem{theorem}{Theorem}[section]
\newtheorem{lemma}[theorem]{Lemma}
\newtheorem{corollary}[theorem]{Corollary}

\theoremstyle{definition}

\theoremstyle{definition}
\newtheorem{remark}[theorem]{Remark}

\numberwithin{equation}{section}

\usepackage{amsmath,amsfonts,amssymb}
\usepackage{graphicx}
\usepackage[round]{natbib}
\usepackage{xcolor}
\usepackage{verbatim}
\usepackage{paralist}
\usepackage{url}

\renewcommand{\le}{\leqslant}
\renewcommand{\ge}{\geqslant}
\renewcommand{\leq}{\leqslant}
\renewcommand{\geq}{\geqslant}

\newcommand{\tmod}{\ \mathsf{mod}\ }

\newcommand{\real}{\mathbb{R}}
\newcommand{\ints}{\mathbb{Z}}
\newcommand{\natu}{\mathbb{N}}
\newcommand{\ftwo}{\mathbb{F}_2}

\newcommand{\bszero}{\boldsymbol{0}}
\newcommand{\bsone}{\boldsymbol{1}}

\newcommand{\rd}{\,\mathrm{d}}
\newcommand{\mrd}{\mathrm{d}}
\newcommand{\dunif}{\mathbb{U}}

\newcommand{\rqmc}{\mathrm{RQMC}}
\newcommand{\simiid}{\stackrel{\mathrm{iid}}\sim}

\newcommand{\e}{\mathbb{E}}
\newcommand{\var}{\mathrm{Var}}

\newcommand{\one}{\mathbf{1}}

\newcommand{\norm}[1]{\Vert #1 \Vert_1}
\newcommand{\Lnorm}{\norm{L}}
\newcommand{\Lknorm}{\norm{L_k}}
\newcommand{\kinfty}{{L_k(1)}}
\newcommand{\supC}{\overline{C}}

\newcommand{\giv}{\!\mid\!}

\newcommand{\tran}{\mathsf{T}}
\newcommand{\phe}{\phantom{=}}
\newcommand{\lset}{\mathcal{L}}

\newcommand{\walk}{\mathrm{wal}_k}

\newcommand{\med}{\mathrm{med}}
\newcommand{\daw}{\mathrm{Daw}}
\newcommand{\cp}{\mathcal{P}}
\begin{document}

\title{Super-polynomial accuracy of one dimensional randomized nets
using the median-of-means}

%    author one information

\author{Zexin Pan\\Stanford University
  \and
  Art B. Owen\\Stanford University}

% \author{Zexin Pan}
% \address{Department of Statistics, Stanford University}
% \email{zep002@stanford.edu}
% \thanks{}

% %    author two information
% \author{Art B. Owen}
% \address{Department of Statistics, Stanford University}
% \email{owen@stanford.edu}
% \thanks{}

% %    \subjclass is required.
% \subjclass[2010]{Primary 65D30, 05A17 }

 \date{July 2022}

%\dedicatory{}

%    Abstract is required.
\maketitle
\begin{abstract}
Let $f$ be analytic on $[0,1]$ with $|f^{(k)}(1/2)|\leq A\alpha^kk!$ for some constants $A$ and $\alpha<2$
and all $k\ge1$.
We show that the median estimate of $\mu=\int_0^1f(x)\rd x$ under random linear scrambling
with $n=2^m$ points converges at the rate $O(n^{-c\log(n)})$ for any
$c< 3\log(2)/\pi^2\approx 0.21$. We also get a super-polynomial convergence rate for
the sample median of $2k-1$ random linearly scrambled estimates, when $k/m$ is bounded away from zero.
When $f$ has a $p$'th derivative that satisfies a $\lambda$-H\"older condition
then the median-of-means has error $O( n^{-(p+\lambda)+\epsilon})$
for any $\epsilon>0$, if $k\to\infty$ as $m\to\infty$.
The proof techniques use methods from analytic combinatorics that have not
previously been applied to quasi-Monte Carlo methods, most notably an
asymptotic expression from Hardy and Ramanujan on the number of
partitions of a natural number.
\end{abstract}

\maketitle

\section{Introduction}

In this paper we introduce and study a median-of-means approach
to randomized quasi-Monte Carlo (RQMC) sampling.  Specifically, for
$f:[0,1]\to\real$ we let $\hat\mu_r$ for $r=1,\dots,2k-1$
be independent estimates of $\mu=\int_0^1f(x)\rd x$
computed using the random linear scrambling of \cite{mato:1998}
applied to a $(0,m,1)$-net in base $2$ and our estimate of $\mu$
is $\hat \mu^{(k)} = \med( \hat\mu_1,\dots,\hat\mu_{2k-1})$.
We find for some infinitely differentiable integrands,
that this median of means approach converges faster
than any polynomial rate in $n=2^m$.
By this we mean that for some $c>0$ the probability of an error
larger than $n^{-c\log(n)}$ approaches zero as
the number of sampled points $n=2^m\to\infty$.

A key ingredient in the proofs is the formula
by \cite{hard:rama:1918} for the number  $p(n)$
of ways to partition the natural number $n$
into a sum of natural numbers.  Their
formula for this is
$$p(n) \sim\frac1{n4\sqrt{3}}\exp\Bigl( \pi \Bigl(\frac{2n}3\Bigr)^{1/2}\Bigr).
$$
We believe that this use of analytic combinatorics in RQMC is new
and we expect further connections to develop.

There have been several recent results on super-polynomial convergence
for quasi-Monte Carlo (QMC).
\cite{suzu:2017}, working in a weighted space of infinitely differentiable
functions on $[0,1]^d$, proved the existence of digital nets with
worst case error $C(d)\exp( - c(d)\log(n)^2)$.  Under further
conditions on the weights defining the space, a dimension free
worst case error $C\exp(-c\log(n)^p)$ holds for some
$1<p<2$.
\cite{dick:goda:suzu:tosh:2017}
give a construction of a super-polynomially
convergent method.  At a cost of $O(nd\log(n)^2)$ they
use a component-by-component construction to get
dimension-independent super-polynomial convergence
using interlaced polynomial lattice rules.
These are higher order digital nets.  A higher order digital
net can attain an error of $\tilde O(n^{-\alpha})$
when the integrand's mixed partial derivatives of total order
up to the integer $\alpha\ge1$ are all in $L^2[0,1]^d$
\citep{dick:2011}.
Here $\tilde O$ means that logarithmic factors are not shown.
Under scrambling, \cite{dick:2011} shows that the
root mean squared error (RMSE) is $\tilde O(n^{-\alpha-1/2})$.
To obtain super-polynomial convergence
\cite{dick:goda:suzu:tosh:2017}
must let the order of their higher-order digital nets increase with $n$.
The median of means formulation allows one to use ordinary
scrambled Sobol' points though in some uses we must
take a median of a (slowly) growing number $k$ of replicates.
The LatNet builder tool of \cite{lecu:etal:2022} constructs
QMC and RQMC point sets using some random searches.
Those searches seek to optimize a figure of merit (FOM)
that quantifies worst case error over a class of integrands.
For the precise definitions of each FOM, see that paper.
Figure 1 there shows some examples where the median
FOM shows curvature on a log-log scale for dimension $d=6$.
This is consistent with super-polynomial accuracy, though they present
a median FOM instead of the FOM of a median estimate.

The algorithm we study here provides another approach.
We will see that when $f$ is smooth, most of the randomized
net estimates are very close to the true value, for large $m$.
The variance is dominated by
a relatively small number of bad outcomes.  By taking the
median of a number of independent estimates we can
reduce the impact of the few bad outcomes.  Each
RQMC estimate is a mean of function evaluations.
Then our combined estimates are a median-of-means.

Median-of-means algorithms have many uses in
theoretical computer science, though the means
used there have not usually been based on RQMC.
See for example, \cite{jerr:vali:vazi:1986} and \cite{lecu:lera:2020}.
\cite{kunsch2019solvable} present several uses of
the median of means in two stage numerical integration
algorithms and they give further references to the literature.
Since our preprint appeared, there has been further
work on median methods for QMC by \cite{goda:lecu:2022:tr}. They choose rank
one lattice generating vectors completely at random for integration problems in Korobov spaces and they choose polynomial lattice rules randomly for some weighted Sobolov spaces.  By taking the median of a number of such randomly
generated estimates they attain the best convergence rates
possible for the smoothness levels they study and they are able to avoid complicated parameter searches.
 \cite{hofstadler2022consistency} use median of means
 to get some strong laws of large numbers for integration
 methods. \cite{gobe:lera:meti:2022:tr} use median of
 means to get robust RQMC estimates.

An outline of this paper is as follows.
Section~\ref{sec:background} provides some notation
and definitions of the scrambling we use and the
resulting estimates.
One key quantity is a scrambling matrix $M$
with $m$ columns and entries in $\{0,1\}$.
The accuracy of RQMC is limited by phenomena
where for some non-empty $L\subset \natu$, the
rows of $M$ for $\ell\in L$ sum to $\bszero$ in $\ftwo^m$.
Section~\ref{sec:bottleneck} explains this bottleneck
to convergence and Theorem~\ref{thm:errordecomposition} there
writes the RQMC error as a sum of random variables,
one for each problematic subset $L$.
This section is less technical than the later ones with our proofs.
Section~\ref{sec:numerical} has a one dimensional numerical example. We see superlinear convergence
for the median of $11$ RQMC replicates, up to a point. The RMSE reaches an asymptote for a Sobol' sequence
computed to $32$ bits. Switching to a $64$ bit computation, the superlinearity continues to some
higher sample sizes.  There is also a six dimensional example, where the standard deviation of median estimates drops faster than that of mean estimates.
Section~\ref{sec:thepopulationmedian}
studies the population median of scrambled
nets for dimension $d=1$.  This is the median of the distribution
of the RQMC estimate.
Theorem~\ref{thm:convergencerate} establishes
a super-polynomial rate for that quantity
for certain infinitely differentiable functions.
A critical step there is to bound the number
of nonempty subsets $L\subset\natu$ that
have a small value of $\sum_{\ell\in L}\ell$. We do
this using combinatorial results including
the one by \cite{hard:rama:1918}.
The comprehensive reference is \cite{flaj:sedg:2009}.
Theorem~\ref{thm:samplemedian} provides super-polynomial
convergence for the median of $2k-1$ independently
generated RQMC estimates.
Section~\ref{sec:finitediff} considers the case where the $p$'th
derivative of $f$ satisfies a $\lambda$-H\"older condition
for $0<\lambda\le 1$.
Theorem~\ref{thm:Cpconvergencerate} there bounds the
probability that the error is much more than $n^{-p-\lambda}$
with corollaries showing super-polynomial convergence for the population median
and the median of $2k-1$ independent estimates.

We close this section with a few contextual remarks.
In the one dimensional setting, there are already
very accurate integration rules for extremely
smooth integrands \citep{davrab}.  The RQMC method here has
an advantage in being an equally weighted average
of the $n$ function values, instead of
having large weights of both positive and negative signs.
The one dimensional case will take on greater interest if
the findings and proofs in this article can be generalized to
$d\ge1$. The multidimensional example in Section~\ref{sec:numerical}
is therefore encouraging
as are the empirical results in \cite{lecu:etal:2022}.

One of the original motivations for RQMC was to get
error estimates.  It later emerged that randomizing QMC
can also increase accuracy \citep{smoovar}.  Error estimation for a median
of independently sampled means is more complicated than
for a mean of such means.  We can readily get a nonparametric
confidence interval for the population median of the $\hat\mu_r$,
using the binomial distribution
because the true median $\theta$ satisfies
$\Pr( \hat\mu_r<\theta)=1/2$.
However, the quantity of most direct interest is $\e(\hat\mu_r)$
not $\med(\hat\mu_r)$.

We had initially considered the case where instead of
a random linear scramble we had taken
a completely random generator matrix
with all entries independent and identically (IID)  $\dunif\{0,1\}$ random variables.
A similar result holds: the median estimate converges
with super-polynomial accuracy for certain infinitely differentiable
$f$, though of course the bad outcomes can be even
worse.  For instance, there is a $2^{-m^2}$ probability
that the upper $m\times m$ submatrix of the generator matrix
is all zeros. Then all $n=2^m$ RQMC points would lie in
the interval $[0,1/n]$ and the resulting error would generally fail to vanish
as $n\to\infty$.

\section{Notation and background}\label{sec:background}

We study the random linear scrambling of \cite{mato:1998}
including a digital shift, in one dimension. Our focus
is on base $2$ apart from a few remarks later.
For an integrand $f:[0,1]\to\real$ we will
estimate $\mu=\int_0^1f(x)\rd x$ assumed to exist
by $\hat\mu=(1/n)\sum_{i=0}^{n-1}f(x_i)$ for carefully
chosen points $x_i\in[0,1]$.

We use $\natu$ for the set of positive integers.
For $m\in\natu$, we let $[m]$ denote the set $\{1,\ldots,m\}$
and for $n\in\natu$ we let $\ints_n$ denote
the set $\{0,1,\dots,n-1\}$.
We investigate a scrambled digital net of $n=2^m$
points $x_i\in[0,1)$ for $i=0,1,\dots,n-1$.

We will make frequent use of
sets $L\subset\natu$ of finite cardinality.
We write $|L|$ for their cardinality
as well as $\Vert L\Vert_1=\sum_{\ell\in L}\ell$
and some additional notation about these
sets $L$ will be introduced as needed.
For a matrix $M$ we use $M(L,:)$ to
denote the submatrix whose row indices are in $L$
and to extract a single row we
write $M(\ell,:)$ instead of $M(\{\ell\},:)$.

The indicator function of the event $E$ is
sometimes written $\bsone\{E\}$.  This quantity
takes the value $1$ when $E$ holds and $0$ otherwise.

%\subsection{The unrandomized net}
We assume throughout that $C\in\{0,1\}^{m\times m}$
is a nonrandom matrix that is of full rank $m$
over $\ftwo$, that is, it has full rank in arithmetic modulo 2.
This matrix $C$ defines the `unscrambled' version
of our QMC points which will be a $(0,m,1)$-net in base $2$.
For instance $C$ could be the $m\times m$ identity matrix
as it would be for the van der Corput points.

For $i\in\ints_{2^m}$ we let $\vec{i} = (i_1,i_2,\dots,i_m)^\tran$
where $i = i_1+2i_2+4i_3+\cdots+2^{m-1}i_m$.
For $a=\sum_{k=1}^m a_k2^{-k}\in[0,1)$ we let
$\vec{a} = (a_1,a_2,\dots,a_m)^\tran$.
These two definitions intersect only for $\{0\}$
where they both yield $\bszero$.
The representation for $a\in[0,1)$ can be taken
to any finite number $E\ge m$ of bits, that we denote by $\vec{a}[E]$
when we need to specify the precision.
When $a$ has two base $2$ representations, we work with
the one that has finitely many nonzero bits.
The points of the unscrambled net are given by
$a_i\in  \{k/2^m\mid k\in\ints_{2^m}\}\subset  [0,1)$ that satisfy
$$\vec{a}_i=C\vec{i}\quad\text{for $i\in \ints_{2^m}$}$$
so that
$$a_i = \sum_{k=1}^m 2^{-k}a_{ik}
\quad\text{for bits}\quad a_{ik}=\sum_{j=1}^kC_{kj}i_j\tmod 2.
$$
In our presentation below we will omit noting that
bitwise arithmetic is done modulo two, when that is
clear from context.

%\subsection{Randomization}
To scramble the points, we introduce a random
matrix $M\in\{0,1\}^{E\times m}$
for $E\ge m$.
The upper triangular elements
of $M$ are all $0$, the diagonal elements of $M$
are all $1$, and the elements below the diagonal
are IID $\dunif\{0,1\}$.

We also introduce a random digital shift
$D=\sum_{k=1}^\infty 2^{-k}D_k$ with bits ${D}_k$ that are independent $\dunif\{0,1\}$ variables
independent of $M$. Note that $\Pr(0\le D<1)=1$.
The random digital shift serves to make the estimates
$\hat\mu$ unbiased estimates of $\mu$ and
our proofs require that property.

For $i\in\ints_{2^m}$, linearly scrambled points $x_i$
to precision $E$ without a digital shift
are defined by
\begin{equation}\label{eqn:xequalMCiwithoutD}
    \vec{x}_i=\vec{x}_i[E]=M\vec{a}_i
%+\vec{D}
=M C\vec{i}%+\vec{D}
\end{equation}
for bits
$$a_{ik}=\sum_{j=1}^k C_{kj}i_j
\quad\text{and}\quad
x_{ik}=\sum_{j=1}^k M_{kj}a_{ij}.$$
When we add the random digital shift,
we randomize infinitely many bits.
We define scrambling of precision $E$ to mean
that $x_i$ has bits
\begin{equation}\label{eqn:xequalMCiplusD}
x_{ik} =x_{ik}[E]=
\begin{cases}
\sum_{j=1}^k M_{kj}a_{ij} + D_k,& k\le E\\
D_k, & k>E.
\end{cases}
\end{equation}
We will use the term `random linear scrambling' to refer
to linear scrambling that includes the digital shift.
This usage is common.  Another usage
calls that affine scrambling
with linear scrambling excluding the digital shift.

%Def: scrambling of precision $E$ (denoted as $\hat{\mu}_E$) means first $E$ digits given by $\vec{x}_i=M\vec{a}_i+\vec{D}=M C\vec{i}+\vec{D}$ and further digits being completely random.

Let $f:[0,1]\to\real$. We need to specify the precision of our
estimates and to do this, we define
$$
\hat\mu_E = \hat\mu_E(f)=\frac1n\sum_{i=0}^{n-1} f(x_i)
$$
where the bits of $x_i$ are given by~\eqref{eqn:xequalMCiplusD}.
When $f$ is continuous on $[0,1]$, define
$$\hat{\mu}_\infty=\lim_{E\to\infty}\hat{\mu}_E.$$
Later when we replicate these quantities, the replicates
will be denoted $\hat\mu_{E,r}$ and $\hat\mu_{\infty,r}$.
We let $\omega_f(t)$ denote the modulus of continuity of $f$ over $[0,1]$.
Later,
$\omega_f(1)$ will be a convenient shorthand for $\sup_{0\le x\le1}f(x)-\inf_{0\le x\le 1}f(x)$.

\begin{lemma}\label{lem:precisionerror}
For any $M\in\{0,1\}^{\infty\times m}$ and $D\in[0,1)$
$$|\hat\mu_\infty-\hat\mu_E|\leq \omega_f\Bigl(\frac{1}{2^E}\Bigr)$$
where $\hat\mu_E$ is constructed using the first $E\ge m$ rows of $M$.
\end{lemma}
\begin{proof}
Let $x_i[E]$ be $x_i$ under scrambling with precision $E$ and $x_i[\infty]$ be $x_i$ under scrambling in the infinite precision limit.
For any given $M$ and $D$ in random linear scrambling, $x_i[E]$ has the same first $E$ bits as $x_i[\infty]$, so
$$\bigl|x_i[E]-x_i[\infty]\bigr|\leq \sum_{k=E+1}^\infty\frac{1}{2^k}|x_{ik}[E]-x_{ik}[\infty]|\leq \frac{1}{2^E},$$
where $k$ indexes the bits of $x_i[E]$ and $x_i[\infty]$.
Hence
$$|\mu_\infty-\mu_E|\leq \frac{1}{n}\sum_{i=0}^{n-1} |f(x_i[E])-f(x_i[\infty])|\leq \omega_f\Bigl(\frac{1}{2^E}\Bigr).
\qedhere$$

\end{proof}

The main object of our study is the median of $2k-1$ independently sampled replicates of a randomized QMC algorithm on $m$ points. We may take $k$ to be a function of $m$.
We write $k=\Omega(m)$ to mean that $\lim\inf_{m\to\infty}k(m)/m>0$
and similarly $k=\Omega(m^2)$ means that $\liminf_{m\to\infty} k(m)/m^2>0$.  In practice $k$ would be non-decreasing in $m$
though our results do not require this.

\section{A bottleneck in convergence}\label{sec:bottleneck}

It is well known that the variance of $\hat\mu$ under
nested uniform scrambling attains $O(n^{-3})$ convergence when $d=1$
and $f'\in C[0,1]$, a great improvement upon the $O(n^{-1})$ rate of naive Monte Carlo.
Corollary 3.8 of \cite{wiar:lemi:dong:2021} shows that random linear scrambling with a digital shift has the same variance
as nested uniform scrambling for $(0,m,1)$-nets.  Increased smoothness does not improve this
rate outside of trivial settings with zero variance.
Here we give a simple argument to illustrate that limitation.  Understanding such bounds
leads us to an expression for the integration error
below, on which we base our study of medians.

If $n=2^m$ and $M(m+1,:)$ happens to be $\bszero$, then by the relationship $x_{i,m+1}=\sum_{j=1}^{m+1} M_{m+1,j}a_{ij}+D_{m+1}$, we immediately see that $x_{i,m+1}=D_{m+1}$ for all $i$.
Geometrically, this means for each interval $[i/n,(i+1)/n)$, the samples are either all in the left half interval (if $D_{m+1}=0$) or all in the right half interval (if $D_{m+1}=1$). If we assume for simplicity that $D_{m+1}=1$ and the scrambling has precision $m+1$, then each sample is actually uniform on the right half of the interval it lands in and we can approximate the error
by its expectation given that $M(m+1,:)=\bszero$ and $D_{m+1}=1$ as:
\begin{align*}
    \hat{\mu}_{m+1}-\mu&\approx\sum_{i=0}^{n-1}
    \biggl(2\int_{\frac{i+0.5}{n}}^{\frac{i+1}{n}}f(x)\rd x-\int_{\frac{i}{n}}^{\frac{i+1}{n}}f(x)\rd x\biggr) \\
   & \approx \sum_{i=0}^{n-1}\int_{\frac{i+0.5}{n}}^{\frac{i+1}{n}}f\Bigl(\frac{i+0.5}{n}\Bigr)+f'\Bigl(\frac{i+0.5}{n}\Bigr)\Bigl(x-\frac{i+0.5}{n}\Bigr)\rd x\\
   &\phe -\sum_{i=0}^{n-1}\int_{\frac{i}{n}}^{\frac{i+0.5}{n}}f\Bigl(\frac{i+0.5}{n}\Bigr)+f'\Bigl(\frac{i+0.5}{n}\Bigr)\Bigl(x-\frac{i+0.5}{n}\Bigr)\rd x\\
   &=\frac{1}{8n^2}\sum_{i=0}^{n-1} f'\Bigl(\frac{i+0.5}{n}\Bigr).
\end{align*}

If instead $D_{m+1}=0$, then
all the $x_i$ fall in the left half interval and the expected error is like that above, but with the opposite sign. Hence the conditional expectation of $|\hat{\mu}_{\rqmc}-\mu|$ cannot 
be of lower order than $n^{-1}$
when $M(m+1,:)=\bszero$. Because each entry of $M(m+1,:)$ is independently 0 or 1 with equal probability, $\Pr(M(m+1,:)=\bszero)=2^{-m}$ and those rare
outcomes alone make $\var(\hat{\mu}_{\rqmc})$ at least of order $2^{-m}(n^{-1})^2=n^{-3}$.
Theorem~\ref{thm:errordecomposition} below makes the above reasoning rigorous.

The main takeaway is that the rare event $M(m+1,:)=\bszero$ makes a major contribution to the variance. Curious readers may ask what happens if we explicitly avoid the event $M(m+1,:)=\bszero$. This is indeed what is done in the affine striped matrix (ASM) scrambling from \cite{altscram}. For  base 2, the matrix $M$ of ASM scrambling is nonrandom and described by $M_{kj}=1$ for $k\geq j$ and $M_{kj}=0$ for $k<j$. Therefore $M(m+1,:)=\bsone$ and ASM scrambling is able to attain the $\var(\hat\mu)=O(n^{-4})$ convergence rate when $f''$ is bounded on $[0,1)$ \citep[Proposition 3.7]{altscram}.

A similar question arises: can ASM scrambling converge faster than $O(n^{-4})$ under stronger smoothness assumptions? The answer is again no. Assume for simplicity that $D_{m+1}=D_{m+2}=0$ and that the scrambling has precision $E=m+2$. Because $M(m+1,:)=M(m+2,:)$,
$$x_{i,m+1}=\sum_{j=1}^{m+1} M_{m+1,j}a_{ij}=\sum_{j=1}^{m+1} M_{m+2,j}a_{ij}=x_{i,m+2}.$$
Now within  each interval $[i/n,(i+1)/n)$, the sampling is either uniform in the leftmost quarter $[i/n,(i+0.25)/n)$ or uniform in the rightmost quarter $[(i+0.75)/n,(i+1)/n)$.
Suppose without loss of generality that the sampling
for interval $i$ is in the rightmost quarter.
Then as in the analysis of random linear scrambling, we can approximate the integration error over $[i/n,(i+1)/n)$ by
\begin{align*}
    &4\int_{\frac{i+0.75}{n}}^{\frac{i+1}{n}}f(x)\rd x-\int_{\frac{i}{n}}^{\frac{i+1}{n}}f(x)\rd x\\
    &\approx 4\int_{\frac{i+0.75}{n}}^{\frac{i+1}{n}}f\Bigl(\frac{i+0.5}{n}\Bigr)+f'\Bigl(\frac{i+0.5}{n}\Bigr)\Bigl(x-\frac{i+0.5}{n}\Bigr)+\frac{1}{2}f''\Bigl(\frac{i+0.5}{n}\Bigr)\Bigl(x-\frac{i+0.5}{n}\Bigr)^2\rd x\\
   &\phe-\int_{\frac{i}{n}}^{\frac{i+1}{n}}f\Bigl(\frac{i+0.5}{n}\Bigr)+f'\Bigl(\frac{i+0.5}{n}\Bigr)\Bigl(x-\frac{i+0.5}{n}\Bigr)+\frac{1}{2}f''\Bigl(\frac{i+0.5}{n}\Bigr)\Bigl(x-\frac{i+0.5}{n}\Bigr)^2\rd x\\
   &=\frac{3}{8n^2}f'\Bigl(\frac{i+0.5}{n}\Bigr)+\frac{1}{32n^3}f''\Bigl(\frac{i+0.5}{n}\Bigr).
\end{align*}
When sampling for observation $i$ is in the leftmost quarter
the approximate error as above is
$$-\frac{3}{8n^2}f'\Bigl(\frac{i+0.5}{n}\Bigr)+\frac{1}{32n^3}f''\Bigl(\frac{i+0.5}{n}\Bigr).$$

The $f'$ terms each contribute an error of $O(n^{-2})$.
The sign of the $f'$ terms depends on the nonrandom $a_i$.
Carefully chosen $a_i$ could possibly bring cancellation among the $f'$ terms, leaving a total error $o(n^{-2})$ from the $f'$
terms. However, no such cancellation is possible for the $f''$ terms.  Therefore, if we did manage to cancel the
$f'$ terms we would still have an error
\begin{align*}
    \hat{\mu}_{m+2}-\mu
   & \approx \frac{1}{32n^3}\sum_{i=0}^{n-1} f''\Bigl(\frac{i+0.5}{n}\Bigr).
\end{align*}
This implies that $|\hat{\mu}_{\rqmc}-\mu|$ is at least of order $n^{-2}$,
and so $\var(\hat{\mu}_{\rqmc})$ cannot converge faster than $O(n^{-4})$. Notice that in this case $M$ is nonrandom, so we do not need to multiply
that squared error by an event probability like  $2^{-m}$
as we did in the previous example.

One may summarize from the above heuristic reasoning that whenever a set $L$ of rows of $M$
satisfies $\sum_{\ell\in L}M(\ell,:)=\bszero$, there is an associated error of order $2^{-\sum_{\ell\in L}\ell}$. This is indeed true by Theorem~\ref{thm:errordecomposition} below.
Before stating that theorem we introduce some notation.
Let
$$
\lset = \{ L\subset\natu\mid 0<|L|<\infty\}.$$
Each $L\in \lset$ identifies a set of row indices for $M$. Each of these finite non-empty subsets of natural numbers potentially
contributes an error that scales like $2^{-\Lnorm}$
where $\Lnorm = \sum_{\ell\in L}\ell$.
We also use $\norm{D(L)}=\sum_{\ell\in L }D_\ell$.
This quantity will appear as the exponent of $-1$ where
only its value modulo two matters.

\begin{theorem}\label{thm:errordecomposition}
Let $f$ be analytic on $[0,1]$ with $|f^{(k)}(1/2)|\leq A\alpha^kk!$ for some constant $A$, some $\alpha<2$ and all $k\in\natu$.
If $C\in\{0,1\}^{m\times m}$ is nonsingular, then
\begin{align}\label{eqn:errordecomposition}
\hat{\mu}_{\infty}-\mu&=\sum_{L\in \lset}\bsone\Bigl\{\sum_{\ell\in L}M(\ell,:)=\bszero\Bigr\} \, % (-1)^{\norm{D(L)}}
S_L(D) \,2^{-\Lnorm} B_L
%\end{align}
\intertext{where }
S_L(D) & = \prod_{\ell\in L}(-1)^{D_\ell}\label{def:slofd}
\intertext{and scalars $B_L$ from Appendix \ref{app:proof:thm:errordecomposition} satisfy}
    |B_L|&\leq 6A \bigl(|L|\bigr)!\Bigl(\frac{\alpha/2}{1-\alpha/2}\Bigr)^{|L|}.
\label{eqn:BLbound}
\end{align}
\end{theorem}
\begin{proof}
See Appendix \ref{app:proof:thm:errordecomposition}.
\end{proof}

\begin{remark}\label{rmk:fassumption}
Notice that $|f^{(k)}(1/2)|\leq A\alpha^kk!$ for all $k\in\natu$ 
is not really a more stringent assumption than $f$ being analytic on $[0,1]$. To be analytic
on a closed interval requires $f$ to be analytic on some open interval containing it. Then for the Taylor expansion of $f$ centered at $1/2$ to have a radius of convergence larger than $1/2$, it is necessary that $|f^{(k)}(1/2)|\leq A\alpha^kk!$ for some constant $A$ and $\alpha<2$.
\end{remark}

We see that for each $L\in\lset$, the corresponding term in~\eqref{eqn:errordecomposition} contains a factor depending on $M$ times a factor depending on $D$.  It helps that $M$ and $D$ are independent random quantities.
\section{Numerical examples}\label{sec:numerical}

The function $f(x)=x\exp(x)$ has integral
$\mu=1$ over $[0,1]$.
We selected this $f$ because it is
infinitely differentiable as our theory
requires, and
it is not a polynomial and is not symmetric
or antisymmetric.  Those are factors that
might make a function artificially easy to
integrate by a specially tuned numerical method.
There is also no special feature in the function
at values like $1/2$ or more generally integers
divided by a power of $2$ that might confer an advantage for Sobol' points which are generated
in base $2$.

We sampled this function with random linear
scrambling for $0\le m\le 15$.  For this we used the Sobol function in the QMCPy
software of \cite{choi:etal:2021}.
We took the
median of $k=11$ RQMC integral estimates $R=250$
times.

Figure~\ref{fig:smooth}
shows how the RMSE of the median of $11$ RQMC estimates
decreases with $n$ as open circles connected by dashed lines. It appears to decrease
at a superpolynomial rate until it reaches a limit of about $10^{-9}$.  The Sobol' points in QMCPy default to $32$ bits for the linear scramble with the digital shift carried out more bits. Our theory is for infinitely many bits. We redid the computations using $64$ bits for the linear scramble, resulting in the solid points connected by solid lines. With $64$ bits the apparent super-polynomial convergence holds through the entire range of
sample sizes in the figure.

The figure also shows
the RMSE of a single RQMC estimate of which there
were $250\times 11=2750$.  There is a reference
curve at the $n^{-3/2}$ rate interpolating
the value for $n=1$.  A dashed line below that by a factor of $\sqrt{11}$ corresponds to accuracy using an average of $11$ RQMC estimates that could have been done at the same cost as the median of $11$ RQMC estimates.  In the next sections
we prove that the median RQMC estimate
converges at a superpolynomial rate.
We also show that the sample median
of $k$ RQMC estimates attains such
a rate when $k$ grows slowly with $m$.

At tiny sample sizes like $1$, $2$ and $4$ we
see that a mean of $11$ RQMC estimates was more accurate
than a median of $11$ RQMC estimates.   By $n\ge 16$, we see the median doing better than the dotted reference line applicable to the mean of $11$ RQMC estimates.
\begin{figure}[t!]
    \centering
    \includegraphics[width=.9\hsize]{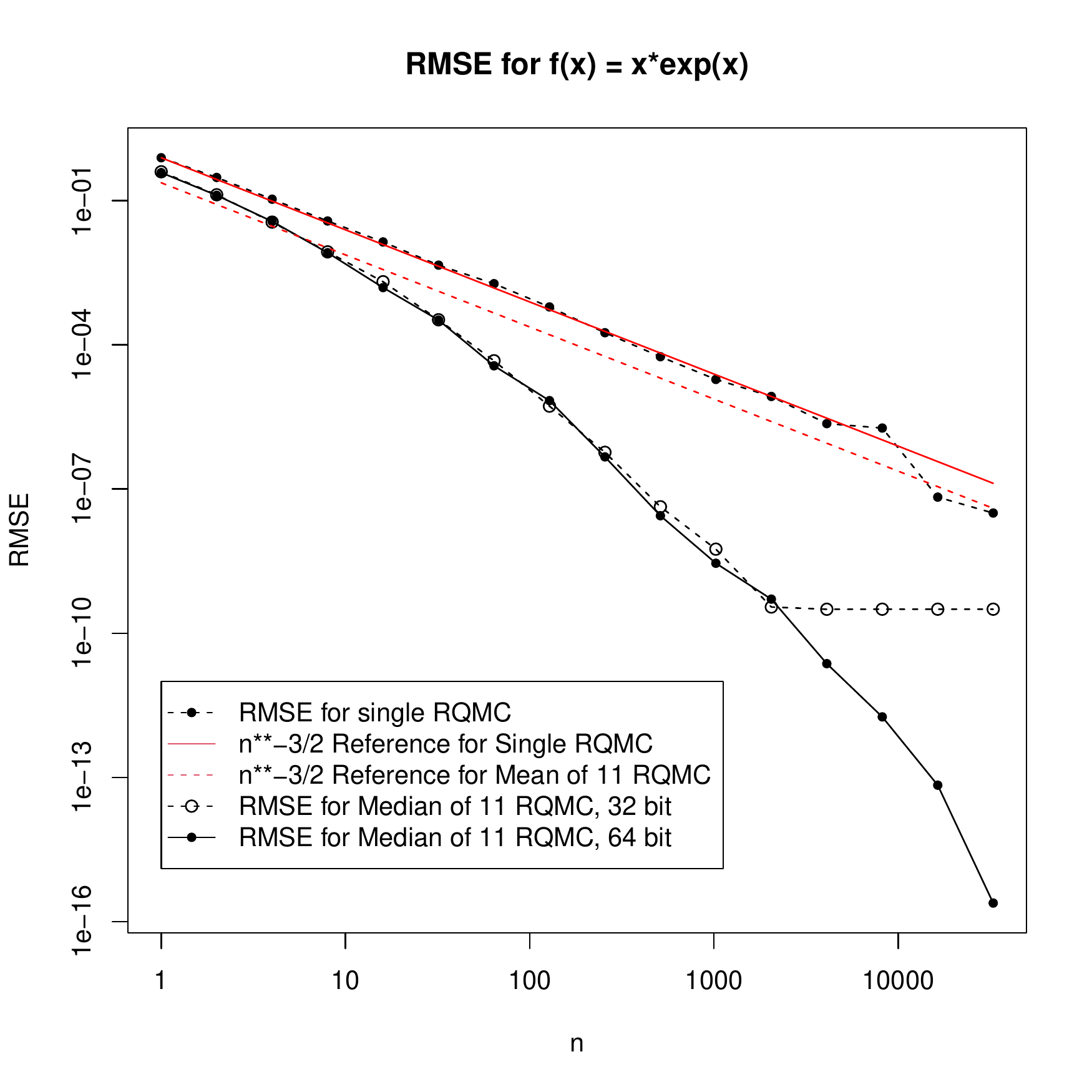}
    \caption{\label{fig:smooth}
    The dashed line with open points shows the RMSE of 250 integral estimates, each of which is the
    median of 11 RQMC estimates. Those computations were done with $M=32$-bit
    Sobol' points.  The solid line with solid points repeats that calculation using 64 bits instead of 32.
     The dashed line with solid points
    connects RMSEs of 2750 RQMC estimates without taking
    a median. The solid reference line is proportional
    to $n^{-3/2}$, running through the plain RQMC value for $n=1$.
    The dashed line is lower by a factor of $\sqrt{11}$ to estimate the RMSE that a mean of $11$ estimates would have.}
\end{figure}

We also investigated a six dimensional function that computes a midpoint voltage
for an output transformerless (OTL)
push-pull circuit.  The function is given
at \cite{surj:bing:2013} which includes a link
to describe the electronics background as well as some code. The results are shown in Figure~\ref{fig:otl}.
We used scrambled Sobol' points from QMCPy. Because the true mean is not known, we plot the standard deviation instead of the RMSE.  While the curve shows an apparent better rate in this multivariate problem it does not account for the bias induced by taking a median instead of a mean. That issue is outside the scope of the present article.

We note in passing that graphical rendering applications of QMC while not having much smoothness can also benefit from using a large number $E$ of bits. See
\cite{kell:2013} for a discussion of QMC for rendering.

\begin{figure}[t!]
    \centering
    \includegraphics[width=.9\hsize]{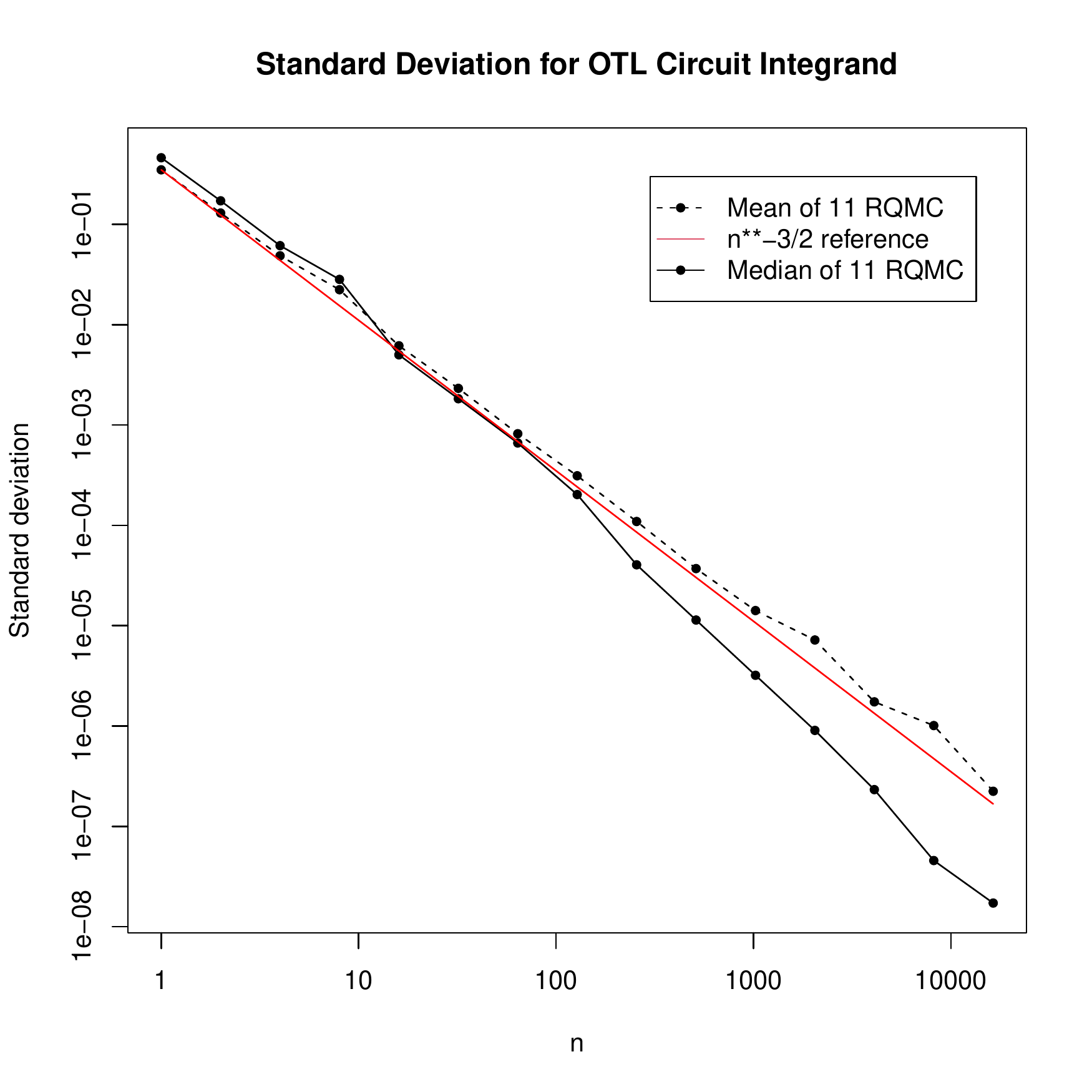}
    \caption{\label{fig:otl}
    The solid line with solid points shows the standard deviation among 100 replicates that each take the median of 11 independent RQMC estimates. The dashed line with solid points has the standard deviation of 1100 replicates divided by $\sqrt{11}$ to reflect the accuracy of a mean of 11 RQMC estimates. The solid reference line is proportional to $n^{-3/2}$ and passes through the point for $n=1$ and the mean of 11 RQMC
    estimates.
        }
\end{figure}
\section{Convergence rate of the median}\label{sec:thepopulationmedian}

As we see in Theorem~\ref{thm:errordecomposition},
sets $L$ with $\sum_{\ell\in L}M(\ell,:)=\bszero$
contribute to the RQMC error and the upper bound
on that contribution contains the factor $2^{-\Lnorm}$, so that sets $L$
with small $\Lnorm$ are of great concern.
In the examples in Section~\ref{sec:bottleneck} we saw that this can be the major source of error in both scrambled nets and ASM sampling.
Is there a way to avoid such bad events? One approach is to redesign the scrambling to avoid $\sum_{\ell\in L}M(\ell,:)=\bszero$ for certain $L$. See for instance the higher order digital nets of \cite{dick:2011} and polynomial lattice rules in \cite{goda:dick:2015}.  Another approach, which is the main focus of this paper, is to take the median instead of the mean of several QMC simulations. Below we show that the median of random linear scrambling with infinite precision converges to $\mu$ at a super-polynomial rate when $f$ satisfies the condition in Theorem~\ref{thm:errordecomposition}.

In random linear scrambling, because $M([m],:)$ is nonsingular and $M(\ell,:)$
for $\ell>m$ has each entry independently $\dunif\{0,1\}$,
\begin{align}\label{eqn:probitszero}
\Pr\biggl(\,\sum_{\ell\in L}M(\ell,:)=\bszero\biggr)=
\begin{cases}
 0, &L\subseteq [m]\\
2^{-m}, & L \nsubseteq [m].
\end{cases}
\end{align}
As a result, the event $\sum_{\ell\in L}M(\ell,:)=\bszero$ is unlikely
to happen for a set $L$ with small $\Lnorm$.   We use the next lemma
to control the number of $L\in \lset$ with small $\Lnorm$.

\begin{lemma}\label{lem:combinatorics} Let $\lambda=3(\log(2))^2/\pi^2\approx 0.146$. Then
\begin{equation}\label{eqn:combinatoricslimit}
\lim_{m\to\infty} \sqrt{m}{2^{-m}} \bigl|\{L\in \lset \mid \Lnorm \leq \lambda m^2\}\bigr|=\frac{3^{1/4}}{2\pi\lambda^{1/4}}.
\end{equation}
Moreover, for $1\leq m\leq 512$,
\begin{align}\label{eqn:combinatorics}
\bigl|\{L\in \lset \mid \Lnorm \leq \lambda m^2\}\bigr|< \frac{0.4\times 2^{m}}{\sqrt{m}}.
\end{align}
\end{lemma}
\begin{proof}
See Appendix \ref{app:proof:lem:combinatorics}.
\end{proof}

The limit in~\eqref{eqn:combinatoricslimit} holds with $m\to\infty$ through real values. Our primary use of it is for integers $m\ge1$ but we will also use it for non-integers.

\begin{remark}
The sequence in equation~\eqref{eqn:combinatoricslimit} is in fact monotonically decreasing for $20\leq m\leq 512$, so one can reasonably guess that the bound in equation~\eqref{eqn:combinatorics} applies to $m>512$ as well, although we do not have a proof for this.
\end{remark}

%\begin{lemma}\label{lem:concentration}
%In random linear scrambling, when $1\leq m\leq 512$,
%$$\Pr\Bigl(\min\Bigl\{\Lnorm\bigm| L\in \lset,\,\sum_{\ell \in L} %M(\ell,:)=\bszero\Bigr\}\leq \lambda m^2\Bigr)< \frac{0.4}{\sqrt{m}}.$$
%In general, the above probability is $O(m^{-1/2})$.
%\end{lemma}
%\begin{proof}
%We take the union bound on all events $\sum_{\ell\in L}M(\ell,:)=\bszero$ with $\Lnorm \leq \lambda m^2$
%Using the union bound we get for $1\leq m\leq 512$
%\begin{align*}
%    &\phe \Pr\Bigl(\exists L\in \lset, \Lnorm \leq  \lambda m^2,\sum_{\ell\in L}M(\ell,:)=\bszero\Bigr)\\
%    &\leq \sum_{\{L\in \lset\mid\Lnorm\leq  \lambda m^2\} } \Pr\Bigl(\,\sum_{\ell\in L}M(\ell,:)=\bszero\Bigr)\\
%   &<\frac{0.4}{\sqrt{m}}
%\end{align*}
%where the last step combines equation~\eqref{eqn:probitszero} and %\eqref{eqn:combinatorics}. When $m>512$, %equation~\eqref{eqn:combinatoricslimit} implies $\bigl|\{L\in \lset \mid %\Lnorm \leq \lambda m^2\}\bigr|=O(2^m/\sqrt{m})$, so the above union bound %shows $\Pr\Bigl(\exists L\in \lset, \Lnorm \leq  \lambda m^2,\sum_{\ell\in %L}M(\ell,:)=\bszero\Bigr)=O(m^{-1/2})$.
%\end{proof}

\begin{lemma}\label{lem:concentration}
In random linear scrambling, there exists a constant $\supC$ such that for all $m\geq 1$ and any $0\leq \epsilon<1$,
$$\Pr\biggl(\min\Bigl\{\Lnorm\bigm| L\in \lset, \sum_{\ell \in L} M(\ell,:)=\bszero\Bigr\}\leq \lambda(1-\epsilon) m^2\biggr)
<  \frac{\supC}{(1-\epsilon)^{{1}/{4}}\sqrt{m}}2^{-{\epsilon m}/{2}}.$$
When $m\leq 512$, we can choose $\supC$ to be $0.4$.
\end{lemma}
\begin{proof}
Equation~\eqref{eqn:combinatoricslimit} implies that there exists a constant $\supC$ such that $\sqrt{m}2^{-m} \bigl|\{L\in \lset \mid \Lnorm \leq \lambda m^2\}\bigr|\leq \supC$ for all $m\geq 1$. We apply this inequality with $m$ replaced by $m\sqrt{1-\epsilon}$ and apply the union bound to all events $\sum_{\ell\in L}M(\ell,:)=\bszero$ with $\Lnorm \leq \lambda(1-\epsilon) m^2$. We then get
\begin{align*}
    &\phe\,\Pr\bigl(\,\exists L\in \lset, \Lnorm \leq  \lambda (1-\epsilon)m^2,\sum_{\ell\in L}M(\ell,:)=\bszero\bigr)\\
        &\leq \sum_{\{L\in \lset\,\mid\,\Lnorm\leq  \lambda(1-\epsilon) m^2\} } \Pr\biggl(\,\sum_{\ell\in L}M(\ell,:)=\bszero\biggr)\\
    &\overset{\mathrm{(i)}}{< } \frac{\supC 2^{m\sqrt{1-\epsilon}}}{(1-\epsilon)^{1/4}\sqrt{m}} 2^{-m}
    \overset{\mathrm{(ii)}}{\leq } \frac{\supC }{(1-\epsilon)^{1/4}\sqrt{m}} 2^{-{\epsilon m}/{2}},
\end{align*}
where (i) follows from equation~\eqref{eqn:probitszero} and (ii) follows from $1-\sqrt{1-\epsilon}\geq \epsilon/2$. When $m\leq 512$, equation~\eqref{eqn:combinatorics} shows we can choose $\supC=0.4$.
\end{proof}
\begin{remark}\label{rmk:concentration}
We will mostly use the above lemma for $\epsilon=0$, in which case
$$\Pr\Bigl(\min\bigl\{\Lnorm\bigm| L\in \lset, \sum_{\ell \in L} M(\ell,:)=\bszero\bigr\}\leq \lambda m^2\Bigr)
< \frac{0.4}{\sqrt{m}}$$
when $1\leq m\leq 512$ and is $O(m^{-1/2})$ as $m\to\infty$.
\end{remark}

We are going to apply Chebyshev's inequality to bound the probability that $\hat{\mu}_{\infty}$ is far from $\mu$.
For that, we first prove that the random sign terms
$S_L(D)=(-1)^{\norm{D(L)}}$
in Theorem~\ref{thm:errordecomposition} are pairwise independent Rademacher (i.e., $\dunif\{-1,1\}$) random variables.

\begin{lemma}\label{lem:uncorrelated}
For $L\in\lset$, let $S_L(D)$
be as in equation~\eqref{def:slofd}.
Then for $L,L'\in\lset$
$\Pr(S_L(D)=1)=
\Pr(S_L(D)=-1)=1/2$ and
for $L\ne L'$
%$$\cov(S_L(D),S_{L'}(D))=
%\begin{cases}
% 1, & L=L'\\
% 0, & L\neq L'.
%\end{cases}
% $$
$$\Pr( S_L(D)=1,S_{L'}(D)=1)=\frac14.$$
\end{lemma}
\begin{proof}
The entries of $D$ are $D_\ell\simiid \dunif\{0,1\}$, so
$(-1)^{D_\ell}\simiid\dunif\{-1,1\}$ and then
$$\e(S_L(D))=\prod_{\ell\in L}\e((-1)^{D_\ell})=0.$$
This combined with $S_L(D)\in\{-1,1\}$
implies that $\Pr(S_L(D)=1)=1/2$.
If $L\ne L'$,
then letting $\triangle$ denote the symmetric difference of sets,
\begin{align}\label{eq:meanrademzero}
\e\bigl( S_L(D) S_{L'}(D)\bigr)
&=\prod_{\ell\in L\triangle L'}\e\bigl((-1)^{D_\ell}\bigr)=0.
\end{align}
Now let $\Pr( S_L(D)=1,S_{L'}(D)=1)=1/4+\delta$.
From the symmetry of Rademacher random variables, we get
$\Pr( S_L(D)=1,S_{L'}(D)=-1)=1/4-\delta$
and $\Pr( S_L(D)=-1,S_{L'}(D)=1)=1/4-\delta$.
Then by subtraction we have $\Pr( S_L(D)=-1, S_{L'}(D)=-1)=1/4+\delta$.
From~\eqref{eq:meanrademzero} we get $\delta=0$
so $\Pr( S_L(D)=1,S_{L'}(D)=1)=1/4$
meaning that $S_L(D)$ and $S_{L'}(D)$ are independent.
\end{proof}

Now we are ready to prove the main theorem concerning the super-polynomial convergence rate of the median of random linear scrambling. Then we will have one corollary for $m\le 512$ and another for $m\to\infty$.

\begin{theorem}\label{thm:convergencerate}
Let the integrand $f$ satisfy the conditions of
Theorem~\ref{thm:errordecomposition}
with constants $A$ and $\alpha$ given there.
Let $\lambda$ be as in Lemma~\ref{lem:combinatorics} and
%let $\theta=\frac{\alpha}{e(2-\alpha)}$.
set $\theta={\alpha}/(e(2-\alpha))$.
Then for any $\eta>0$ and $m\geq 3$,
the random linear scrambling estimate $\hat\mu_\infty$
satisfies
\begin{align}\label{eq:convergencrate}
\Pr\Big(|\hat{\mu}_{\infty}-\mu|>\frac{ A}{\sqrt{\eta}}2^{-\lambda m^2}\sqrt{C_\theta (\theta\sqrt{2\lambda} m)^{2\sqrt{2\lambda} m}+64}\Big)< \eta+\frac{\supC}{\sqrt{m}}
\end{align}
where $C_\theta$  is a positive number depending only on $\theta$,
defined in equation (\ref{eqn:Ctheta}) and $\supC$ is the constant from Lemma~\ref{lem:concentration}. If $m\leq 512$, we can replace $\supC$ by $0.4$. If also
 $m\geq \max((\sqrt{2\lambda}\theta)^{-1}, 3\log(\theta m)+3)$,
then we can replace $C_\theta$ by $3770\max(1,\theta^{-1})$.
\end{theorem}
\begin{proof}
First we condition on $M$ and apply Chebyshev's inequality. For $c>0$
\begin{align*}
   \Pr\bigl(|\hat{\mu}_{\infty}-\mu|>c\giv M\bigr)&\leq \frac{\var(\hat{\mu}_{\infty}-\mu\giv M)}{c^2}.
\end{align*}
By Theorem~\ref{thm:errordecomposition} and Lemma~\ref{lem:uncorrelated},
\begin{align*}
    \var(\hat{\mu}_{\infty}-\mu\giv M)&=
\var\biggl(\,\sum_{L\in \lset}\bsone\biggl\{\sum_{\ell\in L}M(\ell,:)=\bszero\biggr\}
B_{L} S_L(D)2^{-\Lnorm}\giv M\biggr)\\
    &=\sum_{L\in \lset}\bsone\biggl\{\,\sum_{\ell\in L}M(\ell,:)=\bszero\biggr\}
B^2_{L} 4^{-\Lnorm}.
\end{align*}
Let $H$ be the event $\big\{\min\{\Lnorm\mid L\in \lset, \sum_{\ell\in L}M(\ell,:)=\bszero\}> \lambda m^2\big\}$.
By Lemma~\ref{lem:concentration}, $\Pr(H^c)\leq \supC/\sqrt{m}$.  
Conditioning on $H$,
we see that
\begin{align}\label{eqn:sumoverN}
   &\phe\, \e\biggl(\,\sum_{L\in \lset}\bsone\biggl\{\,\sum_{\ell\in L}M(\ell,:)=\bszero\biggr\}
B^2_{L} 4^{-\Lnorm}\!\bigm|\! H\biggr) \nonumber \\
   &=\sum_{L\in \lset}\Pr\biggl(\,\sum_{\ell\in L}M(\ell,:)=\bszero\giv H\biggr) B^2_{L} 4^{-\Lnorm} \nonumber \\
   &=\sum_{N=\lceil \lambda m^2\rceil}^\infty \frac{1}{4^N} \sum_{L\in \lset,\,\Lnorm=N} \Pr\biggl(\,\sum_{\ell\in L}M(\ell,:)=\bszero\!\bigm|\! H\biggr) B^2_{L}.
\end{align}
Now
\begin{equation}\label{eqn:PH}
    \Pr\biggl(\,\sum_{\ell\in L}M(\ell,:)=\bszero\!\bigm|\! H\biggr)\leq
\frac{\Pr(\sum_{\ell\in L}M(\ell,:)=\bszero)}{\Pr(H)}\le\frac{2^{-m}}{\Pr(H)}.
\end{equation}
Furthermore,
$$\Lnorm\geq \sum_{\ell=1}^{|L|}\ell=\frac{|L|(|L|+1)}{2}.$$
So $|L|< \sqrt{2N}$ when $\Lnorm=N$. Let $\lfloor \sqrt{2N}\rfloor$ be the largest integer no larger than $\sqrt{2N}$. Then according to~\eqref{eqn:BLbound}
\begin{align}\label{eqn:blbound}
B_L^2 &\le (6A)^2\bigg((|L|)!
\Bigl(\frac{\alpha/2}{1-\alpha/2}\Bigr)^{|L|}\bigg)^2\notag\\
&<(6A)^2 \bigg((\lfloor \sqrt{2N}\rfloor)!
\Bigl(\frac{\alpha/2}{1-\alpha/2}\Bigr)^{\lfloor \sqrt{2N}\rfloor}\bigg)^2+(6A)^2
\end{align}
where the last inequality uses the fact that factorial (or rather the Gamma function) is logarithmically convex, which implies the maximum is attained at either $|L|=0$ or $|L|=\lfloor \sqrt{2N}\rfloor$.
By Stirling's approximation
%[Robbins, Herbert (1955)],
$$
(\lfloor \sqrt{2N}\rfloor)!< \sqrt{2\pi \lfloor \sqrt{2N}\rfloor}\Bigl(\frac{\lfloor \sqrt{2N}\rfloor}{e}\Bigr)^{\lfloor \sqrt{2N}\rfloor}e^{{1}/{12}}.$$
Applying the above to the bound for $B_L$ in equation \eqref{eqn:BLbound}
we get
\begin{align}\label{eqn:blbound}
B_L^2 &< (6A)^2  2\pi \lfloor \sqrt{2N}\rfloor\Bigl(\frac{\lfloor \sqrt{2N}\rfloor}{e}\Bigr)^{2\lfloor \sqrt{2N}\rfloor}e^{{1}/{6}}
\Bigl(\frac{\alpha/2}{1-\alpha/2}\Bigr)^{2\lfloor \sqrt{2N}\rfloor}+(6A)^2\notag\\
&<2\pi e^{{1}/{6}}(6A)^2 \max(1,\theta^{-1}) \sqrt{2N}\Bigl(\theta\sqrt{2N}\Bigr)^{2 \sqrt{2N}}+(6A)^2.
\end{align}

Next, by Corollary 2 of \cite{bida:2012}
\begin{equation}\label{eqn:Bidar}
    \bigl|\{L\in \lset\mid\Lnorm=N\}\bigr|\leq \frac{\pi \exp\Bigl(\pi\sqrt{\frac{N}{
    3}}\Bigr)}{2\sqrt{3N}}.
\end{equation}
This problem that Bidar studies is different from
that of Hardy and Ramanujan, because the elements
of $L$ must be distinct while Hardy and Ramanujan's
formula involves sums of not necessarily distinct
numbers.
Combining equations~\eqref{eqn:PH},~\eqref{eqn:blbound}, and~\eqref{eqn:Bidar}
\begin{align}
    &\phantom{\leq}\,\sum_{L\in \lset,\Lnorm=N} \Pr\biggl(\,\sum_{\ell\in L}M(\ell,:)=\bszero\giv H\biggr) B^2_{L}\notag\\
    &\leq \frac{2\sqrt{2}\pi^2 e^{1/6}(6A)^2\max(1,\theta^{-1})}{2\sqrt{3}\Pr(H)2^m}(\theta\sqrt{2N})^{2\sqrt{2N}}\exp\Bigl(\pi\sqrt{\frac{N}{3}}\,\Bigr)\notag\\
    &\phe\ +\frac{\pi (6A)^2 }{2\sqrt{3N}\Pr(H)2^m}\exp\Bigl(\pi\sqrt{\frac{N}{
    3}}\Bigr)\notag\\
    &= \frac{\sqrt{2}\pi^2 e^{1/6}(6A)^2\max(1,\theta^{-1})}{\sqrt{3}\Pr(H)2^m}p(N)+\frac{\pi (6A)^2 }{2\sqrt{3N}\Pr(H)2^m}\exp\Bigl(\pi\sqrt{\frac{N}{
    3}}\Bigr),\label{eq:twobigterms}
\end{align}
for
$$
p(N)=(\theta\sqrt{2N})^{2\sqrt{2N}}\exp\Bigl(\pi\sqrt{\frac{N}{3}}\,\Bigr).
$$
%Letting $p$ take positive real valued arguments we
%show next that
%\begin{align}\label{eq:dpbound}
%\frac{\mrd}{\mrd N}\log(p(N)) < \log\Bigl(\frac{4}{1.1}\Bigr)
%\end{align}
%whenever $N\ge \lambda m^2$
%and $m\ge \max((\sqrt{2\lambda}\theta)^{-1}, 3\log(\theta m)+3)$.
%Equation~\eqref{eq:dpbound} follows because ${\mrd}\log(p(N))/{\mrd N}$ is decreasing in $N$ when $\theta\sqrt{2N}\geq 1$. Therefore, \eqref{eq:dpbound} holds for $N\ge \lambda m^2$ if it holds for $N=\lambda m^2$ and $\theta\sqrt{2\lambda m^2}\geq 1$, which are guaranteed by the condition on $m$. 
%Under those conditions
%$$
%\frac{p(N)}{p(\lambda m^2)} \le \Bigl(\frac{4}{1.1}\Bigr)^{N-\lambda m^2}.
%$$

Now define
\begin{equation}\label{eqn:Ctheta}
    C_\theta=\sup_{m\geq 1}\frac{ 12\sqrt{6}\pi^2 e^{1/6}\max(1,\theta^{-1})}{p( \lambda m^2)4^{-\lambda m^2}}\sum_{N=\lceil \lambda m^2\rceil}^\infty \frac{p(N)}{4^N}.
\end{equation}
To see that $C_\theta$ is indeed finite, notice that $(\theta\sqrt{2N})^{2\sqrt{2N}}=\exp(2\sqrt{2N}\log(\theta\sqrt{2N}))$, so $p(N)$ grows at a sub-exponential rate in $N$. More explicitly, for some $\rho>1$ we want to find conditions on $m$ so that $p(N+1)/p(N)<\rho$ for $N\geq \lceil \lambda m^2\rceil$. It is enough to have  ${\mrd}\log(p(N))/{\mrd N}<\log(\rho)$ for $N\geq \lambda m^2$, where we let $p$ take positive real valued arguments. To further simplify the calculation, we assume $m\geq (\sqrt{2\lambda}\theta)^{-1}$ so that $\theta\sqrt{2N}\geq 1$ for $N\geq \lambda m^2$. Then ${\mrd}\log(p(N))/{\mrd N}$ is decreasing in $N$ and we only need to verify that ${\mrd}\log(p(N))/{\mrd N}<\log(\rho)$ at $N=\lambda m^2$. A lengthy but straightforward calculation shows that this holds when
$$ \log(\rho)m > \log(\theta m)\sqrt{\frac{2}{\lambda}}+\frac{2+\log(2\lambda)}{\sqrt{2\lambda}}+\frac{\pi}{6}\sqrt{\frac{3}{\lambda}}.$$
To present the above inequality in a simpler form, we choose $\rho=4/1.1$ and approximate the inequality numerically with a sufficient
condition that $m\ge  3\log(\theta m)+3$. 

In summary, when $m\ge \max((\sqrt{2\lambda}\theta)^{-1}, 3\log(\theta m)+3)$, then $p(N+1)/p(N)<4/1.1$ for $N\geq  \lambda m^2$ and 
\begin{align}\label{eq:get3770}
12\sqrt{6}\pi^2 e^{1/6}
\sum_{N=\lceil \lambda m^2\rceil}^\infty \frac{p(N)/p(\lambda m^2)}{4^{N-\lambda m^2}}&\leq 12\sqrt{6}\pi^2 e^{1/6}\sum_{N=\lceil \lambda m^2\rceil}^\infty 1.1^{\lambda m^2-N}\notag\\
&\leq 3770.
\end{align}
We see that $C_\theta$ is finite and then
\begin{align*}
    &\phe\ \sum_{N=\lceil \lambda m^2\rceil}^\infty \frac{1}{4^N}\frac{\sqrt{2}\pi^2 e^{1/6}(6A)^2\max(1,\theta^{-1})}{\sqrt{3}\Pr(H)2^m}p(N)\\
    &\leq  \frac{C_\theta A^2}{\Pr(H)} \frac{p(\lambda m^2)4^{-\lambda m^2}}{2^m}\\
    &=\frac{C_\theta A^2}{\Pr(H)} 4^{-\lambda m^2}(\theta\sqrt{2\lambda m^2})^{2\sqrt{2\lambda m^2}}\frac{1}{2^m}\exp\Bigl(\pi\sqrt{\frac{\lambda m^2}{3}}\Bigr)\\
    &= \frac{C_\theta A^2}{\Pr(H)} 4^{-\lambda m^2}(\theta\sqrt{2\lambda} m)^{2\sqrt{2\lambda} m}
\end{align*}
where the last equality follows from $\lambda=3(\log(2))^2/\pi^2$.
This bounds the first term in~\eqref{eq:twobigterms} when summed over $N$ as in~\eqref{eqn:sumoverN}.

For the second term, we use the inequality $\sqrt{x+a}\leq \sqrt{x}+{a}/({2\sqrt{x}})$ for $a\geq 0$ with $x=\lambda m^2/3$ and $a = (N-\lambda m^2)/3$
to get
\begin{align}\label{eq:sqrtslopebound}\pi\sqrt{\frac{N}3}
\le\pi\sqrt{\frac{\lambda m^2}3}+\pi\frac{N-\lambda m^2 }{2\sqrt{3\lambda m^2}}.
\end{align}
Then using~\eqref{eq:sqrtslopebound} and
$2^m=\exp(\pi\sqrt{\lambda m^2/3})$ and
the assumption that $m\ge3$,
\begin{align*}
& \sum_{N=\lceil \lambda m^2\rceil}^\infty \frac{1}{4^N} \frac{\pi (6A)^2 }{2\sqrt{3N}\Pr(H)2^m}\exp\Bigl(\pi\sqrt{\frac{N}{
    3}}\Bigr)\\
    &\leq \frac{\pi (6A)^2 }{2\sqrt{3\lambda m^2}\Pr(H)2^m}\exp\Bigl(\pi\sqrt{\frac{\lambda m^2}{3}}\Bigr) 4^{-\lambda m^2} \sum_{N=\lceil \lambda m^2\rceil}^\infty \exp\Bigl(
    \Bigl(\frac{\pi}{2\sqrt{3\lambda m^2}}-\log(4)\Bigr)(N-\lambda m^2)\Bigr)\\
    & \leq \frac{A^2 }{\Pr(H)}4^{-\lambda m^2}\frac{36\pi }{2\sqrt{27\lambda }} \sum_{N=\lceil \lambda m^2\rceil}^\infty \exp\Bigl(
    \Bigl(\frac{\pi}{2\sqrt{27\lambda}}-\log(4)\Bigr)(N-\lambda m^2)\Bigr)\\
    &\leq \frac{64 A^2}{\Pr(H)} 4^{-\lambda m^2}.
\end{align*}

% and the assumption that $m\geq 3$ to get
% \begin{align*}
%   & \sum_{N=\lceil \lambda m^2\rceil}^\infty \frac{1}{4^N} \frac{\pi (3A)^2 }{2\sqrt{3N}\Pr(H)2^m}\exp\Bigl(\pi\sqrt{\frac{N}{
%     3}}\Bigr)\\
%     &\leq \frac{\pi (3A)^2 }{2\sqrt{3\lambda m^2}\Pr(H)2^m}\exp\Bigl(\pi\sqrt{\frac{\lambda m^2}{3}}\Bigr) 4^{-\lambda m^2} \sum_{N=\lceil \lambda m^2\rceil}^\infty \exp\Bigl(
%     \Bigl(\frac{\pi}{2\sqrt{3\lambda m^2}}-\log(4)\Bigr)(N-\lambda m^2)\Bigr)\\
%     & \leq \frac{A^2 }{\Pr(H)}4^{-\lambda m^2}\frac{9\pi }{2\sqrt{27\lambda }} \sum_{N=\lceil \lambda m^2\rceil}^\infty \exp\Bigl(
%     \Bigl(\frac{\pi}{2\sqrt{27\lambda}}-\log(4)\Bigr)(N-\lambda m^2)\Bigr)\\
%     &\leq \frac{16 A^2}{\Pr(H)} 4^{-\lambda m^2}.
% \end{align*}

Using the bounds for both terms
\begin{align}\label{eqn:ConditionalExpectedVariance}
  \nonumber  \e\bigl(\var(\hat{\mu}_{\infty}-\mu\mid M)\giv H\bigr)
    &\leq \sum_{N=\lceil \lambda m^2\rceil}^\infty \frac{1}{4^N} \sum_{L\in \lset,\Lnorm=N} \Pr\biggl(\,\sum_{\ell\in L}M(\ell,:)=\bszero\!\bigm|\! H\biggr) B^2_{L}\\ 
    &\leq \frac{C_\theta A^2}{\Pr(H)} 4^{-\lambda m^2}(\theta\sqrt{2\lambda} m)^{2\sqrt{2\lambda} m}+\frac{64 A^2}{\Pr(H)} 4^{-\lambda m^2}.
\end{align}

Finally,
\begin{align}\label{eqn:Chebshev}
    \Pr(|\hat{\mu}_{\infty}-\mu|>c)&\leq \Pr(|\hat{\mu}_{\infty}-\mu|>c\giv H)\Pr(H)+\Pr(H^c)\nonumber \\
    &\leq \frac{1}{c^2}\e(\var(\hat{\mu}_{\infty}-\mu\giv M)\giv H)\Pr(H)+\Pr(H^c).
\end{align}
The bound~\eqref{eq:convergencrate} follows by choosing
$$c=\frac{ A}{\sqrt{\eta}}2^{-\lambda m^2}\sqrt{C_\theta (\theta\sqrt{2\lambda} m)^{2\sqrt{2\lambda} m}+64}$$ and noting that $\Pr(H^c)< \supC/\sqrt{m}$
by Lemma~\ref{lem:concentration}.
That we can take $\supC=0.4$ for $m\le 512$
follows by Lemma~\ref{lem:concentration}.
That we can take $C_\theta=3770\max(1,\theta^{-1})$
under the given conditions follows by~\eqref{eq:get3770}.
\end{proof}

We can interpret Theorem~\ref{thm:convergencerate}
as placing some control on the probability
that the error $|\hat\mu_\infty-\mu|$ is appreciably
larger than $2^{-\lambda m^2}=n^{-\lambda \log_2(n)}$.
That probability cannot be larger than
$\eta + O(1/\sqrt{m})$ for any $\eta>0$.
The next corollaries show that this provides
some control on the distribution of $|\hat\mu_\infty-\mu|$.
The median of that distribution must converge
rapidly to zero.  Then further below
we translate this property into a property
of the sample median.

\begin{corollary}\label{cor:median}
Under the conditions of Theorem~\ref{thm:convergencerate}
let $\med(\hat\mu_\infty)$ be the median of the distribution
of $\hat\mu_\infty$. Then for $3\leq m\leq 512$
$$|\med(\hat{\mu}_{\infty})-\mu|\leq 2 A 2^{-\lambda m^2}\sqrt{C_\theta (\theta\sqrt{2\lambda} m)^{2\sqrt{2\lambda} m}+64}.$$
\end{corollary}
\begin{proof}
Choose $\eta={1}/{2}-{0.4}/\sqrt{3}$ and apply Theorem~\ref{thm:convergencerate}, we see that 
$$\Pr\Big(|\hat{\mu}_{\infty}-\mu|>2A2^{-\lambda m^2}\sqrt{C_\theta (\theta\sqrt{2\lambda} m)^{2\sqrt{2\lambda} m}+64}\Big)< \eta+\frac{0.4}{\sqrt{3}}=\frac{1}{2}$$
where we have used $\eta={1}/{2}-{0.4}/\sqrt{3}>1/4$ to replace $1/\sqrt{\eta}$ by 2, $m\leq 512$ to replace $\overline{C}$ by $0.4$ and $m\geq 3$ to bound $0.4/\sqrt{m}$ by $0.4/\sqrt{3}$. This conclusion follows once we notice the above probability must exceed $1/2$ if $\med(\hat\mu_\infty)$ falls outside that interval.
\end{proof}

\begin{corollary}\label{cor:superpolynomialrate} For $f$ analytic on $[0,1]$,
$$|\med(\hat{\mu}_{\infty})-\mu|=o(2^{-(\lambda-\epsilon)m^2})$$
for any $\epsilon>0$.
\end{corollary}
\begin{proof}
Remark~\ref{rmk:fassumption} shows such $f$ satisfies the assumption of Theorem~\ref{thm:convergencerate} for some $A$ and $\alpha$, so equation~\eqref{eqn:ConditionalExpectedVariance} shows $\e\bigl(\var(\hat{\mu}_{\infty}-\mu\mid M)\giv H\bigr)=4^{-\lambda m^2+O(m\log (m))}$. 
Let $c=2^{-(\lambda-\epsilon)m^2}$. As in  equation~\eqref{eqn:Chebshev}
\begin{align*}
    \Pr(|\hat{\mu}_{\infty}-\mu|>c)&\leq  \frac{1}{c^2}\e(\var(\hat{\mu}_{\infty}-\mu\giv M)\giv H)\Pr(H)+\Pr(H^c)\\
    &=\frac{4^{-\lambda m^2+O(m\log (m))}}{4^{-(\lambda-\epsilon)m^2}}+O\Bigl(\frac{1}{\sqrt{m}}\Bigr)=o(1)
\end{align*}
where we have used Lemma~\ref{lem:concentration} to bound $\Pr(H^c)$.
The same argument in Corollary~\ref{cor:median} shows $|\med(\hat{\mu}_{\infty})-\mu|<c$ for large enough $m$.
\end{proof}

\begin{remark}\label{rmk:convergencerate}
Because
$$2^{-\lambda m^2}=2^{-\log_2(n)^2{3(\log(2))^2}/{\pi^2}}=
n^{-(3\log(2)/\pi^2)\log(n)}\approx n^{-0.21 \log(n)},$$
Corollary~\ref{cor:superpolynomialrate} shows that the median of $\hat{\mu}_\infty$ converges to $\mu$ faster than any polynomial rate.
\end{remark}

In practice, one can only use finite-precision scrambling and estimate the population median of $\hat\mu$ by the median of a finite number of replicated samples. Below we present results for the sample median.
\begin{theorem}\label{thm:samplemedian}
Suppose the scrambling has precision $E$ and let $\hat{\mu}^{(k)}_{E}$ be the sample median of $2k-1$ independently generated values of $\hat{\mu}_E$. Under the conditions of Theorem~\ref{thm:convergencerate}, for any $\eta>0$ and $m\geq 3$
\begin{align*}
&\phe\, \Pr\Big(|\hat{\mu}^{(k)}_{E}-\mu|>\frac{ A}{\sqrt{\eta}}2^{-\lambda m^2}\sqrt{C_\theta (\theta\sqrt{2\lambda} m)^{2\sqrt{2\lambda} m}+64}+\omega_f\Bigl(\frac{1}{2^{E}}\Bigr)\Big)\\
&<
{2k-1\choose k}\Bigl(\eta+\frac{\supC}{\sqrt{m}}\Bigr)^{k}.
\end{align*}
\end{theorem}
\begin{proof}
Let $\hat\mu_{E,r}$ for $r=1,\dots,2k-1$ be independently sampled
estimates of $\mu$ using linear scrambling of precision $E$
with $n=2^m$.
For each of these, let $\hat\mu_{\infty,r}$ be the
corresponding infinite precision sample value
obtained from the same scrambling matrix $M$
and digital shift $D$ that $\hat\mu_{E,r}$ uses.
The median of the $\hat\mu_{\infty,r}$ is
denoted by $\hat\mu^{(k)}_\infty$.
By Lemma~\ref{lem:precisionerror}, $|\hat{\mu}_{E,r}-\hat\mu_{\infty,r}|
\le\omega_f(2^{-E})$ for $1\le r\le 2k-1$
and so  $|\hat{\mu}^{(k)}_{E}-\hat{\mu}^{(k)}_{\infty}|\leq \omega_f(2^{-E})$.

In order to have
$$
|\hat\mu_\infty^{(k)}-\mu|> \rho :=
\frac{A}{\sqrt{\eta}}2^{-\lambda m^2}\sqrt{C_\theta (\theta\sqrt{2\lambda} m)^{2\sqrt{2\lambda} m}+64}$$
there must be at least $k$ of the $\hat{\mu}_{\infty,r}$
with $|\hat\mu_{\infty,r}-\mu|>\rho$.
By applying Theorem~\ref{thm:convergencerate} to 
each $\hat\mu_{\infty,r}$, we find that
$$
\Pr( |\hat\mu_{\infty,r}-\mu|>\rho)
\le
\eta+\frac{\supC}{\sqrt{m}}.$$
The result follows using the union bound on all ${2k-1\choose k}$ possible
sets of $k$ estimates $\hat\mu_{\infty,r}$ with errors above $\rho$
along with $|\hat{\mu}^{(k)}_{E}-\mu|\leq |\hat{\mu}^{(k)}_{E}-\hat{\mu}^{(k)}_{\infty}|+|\hat{\mu}^{(k)}_{\infty}-\mu|$.
\end{proof}

\begin{corollary}\label{cor:concentration}
Under the conditions of Theorem~\ref{thm:samplemedian}
suppose that $8\leq m\leq 512$ and $E\geq \lceil \lambda m^2\rceil$. Then
$$\Pr\Big(|\hat{\mu}^{(k)}_{E}-\mu|
%>\bigl(5C_\theta A(\theta\sqrt{2\lambda} m)^{\sqrt{2\lambda} m}+\sup_{x\in(0,1)}|f'(x)|\bigr)2^{-\lambda m^2}\Big)
>\Bigl(5 A\sqrt{C_\theta (\theta\sqrt{2\lambda} m)^{2\sqrt{2\lambda} m}+64}+\Vert f'\Vert_\infty\Bigr)2^{-\lambda m^2}\Big)
<\Bigl(\frac{3}{4}\Bigr)^k$$
where $\Vert f'\Vert_\infty =\sup_{0\le x\le 1}|f'(x)|$.
\end{corollary}
\begin{proof}

We begin by noting that ${2k-1\choose k} < 4^k$ holds
for integers $k\ge1$. It holds for $k=1$ and
to complete an induction argument we
find for $k\ge1$ that
$$
{2k+1 \choose k+1}\Bigm/
{2k-1 \choose k} = \frac{(2k+1)(2k)}{k(k+1)}<4.
$$
We choose $\eta=1/25$ in Theorem~\ref{thm:samplemedian}. Then for $512\ge m\ge 8$
$${2k-1\choose k}\Bigl(\eta+\frac{0.4}{\sqrt{m}}\Bigr)^{k}\leq 4^k\Bigl(\frac{3}{16}\Bigr)^k=\Bigl(\frac{3}{4}\Bigr)^k$$
where $m\ge8$ was used to make
$\eta +0.4/\sqrt{m}\le 3/16$.
Now the conclusion follows because $2^{-E}\leq 2^{- \lambda m^2}$ and $\omega_f(t)\leq\Vert f'\Vert_\infty t$.
%(\sup_{x\in(0,1)}|f'(x)|)t$.
\end{proof}

\begin{corollary}
Assume that $E\geq \lceil \lambda m^2\rceil$ and that
$k$ is nondecreasing in $m$. Then
$$\e(|\hat{\mu}^{(k)}_{E}-\mu|^2)\leq
4^{-(1-\frac{4\lambda m}{k+4\lambda m}) \lambda m^2+O(m\log(m))}.$$
In particular, when $k=\Omega(m)$, the MSE of $\hat{\mu}^{(k)}_{E}$ converges to $\mu$ at a super-polynomial rate.
If {further}  $k=\Omega( m^2)$, then
$$\e(|\hat{\mu}^{(k)}_{E}-\mu|^2)\leq 4^{-\lambda m^2+O(m\log(m))}.$$
\end{corollary}
\begin{proof} We first introduce a parameter $0\leq \epsilon<1$ and change the event $H$ we used in the proof of Theorem~\ref{thm:convergencerate} to be $$H=\big\{\min\{\Lnorm\mid L\in \lset, \sum_{\ell\in L}M(\ell,:)=\bszero\}> \lambda (1-\epsilon)m^2\big\}.$$ Then as in equation~\eqref{eqn:Chebshev}, we can choose 
$$c=\frac{ A}{\sqrt{\eta}}2^{-\lambda (1-\epsilon) m^2}\sqrt{C_\theta (\theta\sqrt{2\lambda} m)^{2\sqrt{2\lambda} m}+64}$$
and conclude that $\Pr(|\hat{\mu}_{\infty}-\mu|>c)<\eta+\Pr(H^c)$. By Lemma~\ref{lem:concentration}, 
$$\Pr(H^c)<\frac{\supC}{(1-\epsilon)^{{1}/{4}}\sqrt{m}}2^{-{\epsilon m}/{2}},$$
so we can choose $\eta$ such that
$$\eta+\Pr(H^c)=\frac{2\supC}{(1-\epsilon)^{{1}/{4}}\sqrt{m}}2^{-{\epsilon m}/{2}}.$$
With this choice, $c=2^{-\lambda (1-\epsilon) m^2+O(m\log m)}$ and a similar reasoning as in Theorem~\ref{thm:samplemedian} shows
\begin{align}\label{eqn:MSE}
  \e(|\hat{\mu}^{(k)}_{E}-\mu|^2)&\leq \Pr\Big(|\hat{\mu}^{(k)}_{E}-\mu|>c+\omega_f\Bigl(\frac{1}{2^{E}}\Bigr)\Big)\omega_f(1)^2\\
  &\phe\ +\Pr\Big(|\hat{\mu}^{(k)}_{E}-\mu|\leq c+\omega_f\Bigl(\frac{1}{2^{E}}\Bigr)\Big)\Big(c+\omega_f\Bigl(\frac{1}{2^{E}}\Bigr)\Big)^2 \nonumber \\
   &\leq  \Bigl(\frac{8\supC}{(1-\epsilon)^{1/4}\sqrt{m}}\Bigr)^k2^{-\frac{k\epsilon m}{2}}\omega_f(1)^2+4^{-\lambda(1-\epsilon) m^2+O(m\log(m))}
\end{align}
where we have used $\omega_f\bigl(\frac{1}{2^{E}}\bigr)=O(2^{-\lambda m^2})$ when $E\geq \lceil \lambda m^2\rceil$ and ${2k-1\choose k} < 4^k$ as in Corollary~\ref{cor:concentration}.
%We find that  $\e(|\hat{\mu}^{(k)}_{E}-\mu|^2)$ is no larger than
%\begin{align}\label{eqn:MSE}
%  &\phe\, \Pr\Big(|\hat{\mu}^{(k)}_{E}-\mu|>2^{-\lambda(1-\epsilon) m^2+O(m\log(m))}\Big)
%  \omega_f(1)^2
%  \nonumber \\
%   &\phe\,+\Pr\Big(|\hat{\mu}^{(k)}_{E}-\mu|\leq 2^{-\lambda(1-\epsilon) m^2+O(m\log(m))}\Big)4^{-\lambda(1-\epsilon) m^2+O(m\log(m))} \nonumber \\
%   &\leq  \Bigl(\frac{4\supC}{(1-\epsilon)^{1/4}\sqrt{m}}\Bigr)^k2^{-\frac{k\epsilon m}{2}}
%\omega_f(1)^2 +4^{-\lambda(1-\epsilon) m^2+O(m\log(m))}.
%\end{align}
%The last inequality follows from Lemma~\ref{lem:concentration} and Theorem~\ref{thm:samplemedian} with $m$ replaced by $m\sqrt{1-\epsilon}$.

If we choose $\epsilon={4\lambda m}/({k+4\lambda m})$ so that $2^{-k\epsilon m/2}=4^{-\lambda(1-\epsilon) m^2}$, then
$$\frac{4C}{(1-\epsilon)^{1/4}\sqrt{m}}=\Bigl(\frac{k+4\lambda m}{k}\Bigr)^{1/4}\frac{4C}{\sqrt{m}}\leq (1+4\lambda m)^{1/4}\frac{4C}{\sqrt{m}}=o(1).$$
So the second term in equation~\eqref{eqn:MSE} dominates. When $k=\Omega( m^2)$, we choose $\epsilon=0$ and notice that in this case $(4C/\sqrt{m})^k=o(4^{-\lambda m^2})$.
\end{proof}

Our proof strategy generalizes to digital nets with prime base $b\geq 3$, as we now sketch.
Let $\mathbb{F}^*_b$ be the nonzero elements in field $\mathbb{F}_b$,
let $y_L=(y_{L,1},\dots,y_{L,|L|})$ be a length $|L|$ vector with entries in $\mathbb{F}^*_b$ and let $\mathbb{F}^{*|L|}_b$ be the set of all such length $|L|$ vectors.
One can show that
$$\hat{\mu}_{\infty}-\mu=\sum_{L\in \lset}\sum_{y_L\in \mathbb{F}^{*|L|}_b }\bsone\biggl\{\sum_{\ell\in L}y_{L,\ell}M(\ell,:)=\bszero\biggr\} B_{L,y_L} \zeta_{L,y_L}(D) b^{-\Lnorm}$$
where $\zeta_{L,y_L}(D)$ is a complex number of modulus 1 (actually an integer power of $\exp(2\pi\sqrt{-1}/b)$) and $B_{L,y_L}$ is a constant obeying a bound similar to that in Theorem~\ref{thm:errordecomposition}.
After applying the union bound as in Lemma~\ref{lem:concentration}, one can prove that with high probability all $b^{-\Lnorm}$ are small and the convergence of the median is super-polynomial. However, the constant $c$ in $|\med(\hat{\mu}_{\infty})-\mu|=O(n^{-c\log(n)})$ must be smaller. To see this, notice that each $L$ is associated with $(b-1)^{|L|}$ distinct $y_L$, so
$$|\{y_L\in \mathbb{F}^{*|L|}_b ,L\in \lset \mid \Lnorm = N\}|=\sum_{L\in \lset}\bsone_{\Lnorm = N} (b-1)^{|L|}$$
which is obviously smallest and simplest for $b=2$.
Let $q(N)=\sum_{L\in \lset}\bsone_{\Lnorm = N} $ be the number of $L$ with $\Lnorm=N$.
We know that $$q(N)\sim \frac{C}{N^{3/4}}\exp\Bigl(\pi\sqrt{\frac{N}{3}}\Bigr)$$
for some constant $C$ from VIII.24 of \cite{flaj:sedg:2009}.
Applying the arithmetic versus geometric mean inequality
$$\frac{1}{q(N)}\sum_{L\in \lset}\bsone_{\Lnorm = N} (b-1)^{|L|}\geq (b-1)^{\frac{1}{q(N)}\sum_{L\in \lset}\bsone_{\Lnorm = N}|L|}.$$
Now $q(N)^{-1}\sum_{L\in \lset}\bsone_{\Lnorm = N}|L|$ is the average length of $L$ with $\Lnorm=N$,
and that equals $(2\sqrt{3}\log(2)/\pi)\sqrt{N}+o(\sqrt{N})$
by VII.28 of \cite{flaj:sedg:2009}.
So roughly speaking, $|\{y_L\in \mathbb{F}^{*|L|}_b ,L\in \lset \mid \Lnorm \leq N\}|$ is lower bounded by
$\exp(\kappa_b\sqrt{N})$ for
$$\kappa_b=\pi\sqrt{\frac{1}{3}}+\frac{2\sqrt{3}\log(2)\log(b-1)}{\pi}.$$
By setting $\exp(\kappa_b\sqrt{N})=b^m$, we see that the union bound can at best guarantee no $L$ with
$\Lnorm\leq (\log(b)/\kappa_b)^2 m^2 $ would satisfy $\sum_{\ell\in L}y_{L,\ell}M(\ell,:)=\bszero$ for any $y_L$
and the error bound for $|\hat{\mu}_{\infty}-\mu|$ is at best $O(b^{-(\log(b)/\kappa_b)^2 m^2})=O(n^{-(\log(b)/\kappa_b^2)\log(n)})$.
One can prove that $\log(b)/\kappa_b^2$ is larger for $b=2$
than for any integer $b>2$. To prove this, let $h(b)=\log(b)/\kappa_b^2$ be a function of $b\in[2,\infty)$. The derivative $h'$ is negative for $b>b_*$ with a modest value $b_*$ and we can check all integers in the interval $[2,b_*]$ finding that $b=2$ is the maximizer.
Also $\log(2)/\kappa_2^2=3\log(2)/\pi^2$ is the rate constant that
Remark~\ref{rmk:convergencerate} shows
is attained for the base $2$ case up to an arbitrarily small $\epsilon$.

To be clear, the above heuristic reasoning does not prove the asymptotic convergence rate of the median of random linear scrambling with base $b\geq 3$ is slower. It only suggests the proof strategy used in our paper cannot produce a better upper bound than the base 2 case. Results may change if one can reason more cleverly and replace the union bound with a tighter bound.

\section{Convergence rate under finite differentiability}\label{sec:finitediff}

Although our method is designed for smooth target functions, in applications one may not know the exact smoothness but still want to apply this method. In this section, we show that the median converges at almost the optimal rate for the class of $p$ times differentiable functions $C^p([0,1])$ whose $p$'th derivatives are $\lambda$-Hölder continuous.
The $\lambda$ in H\"older continuity should not be confused with the $\lambda$ from Lemma~\ref{lem:combinatorics}.

First we state the counterpart to Theorem~\ref{thm:convergencerate} in this setting.
By a partition of $[0,1]$ we mean a sequence
of increasing numbers $0=x_0<x_1<\dots<x_N=1$ where $N\in\natu$.
For $0<\lambda\leq 1$ and a function $f$ over $[0,1]$, we use
the $\lambda$-variation measure
$$V_\lambda(f)=\sup_{\cp}\sum_{i=1}^N |x_i-x_{i-1}|\frac{|f(x_i)-f(x_{i-1})|}{|x_i-x_{i-1}|^{\lambda}}$$
from \citet[Chapter 14.4]{dick:pill:2010},
in which $\cp$ is the set of all partitions of $[0,1]$.
Notice that when $\lambda=1$, $V_1(f)$ is the total variation of $f$. If $f$ is $\lambda$-Hölder continuous, namely $|f(x_i)-f(x_{i-1})|\leq C |x_i-x_{i-1}|^{\lambda}$ for some constant $C$, then $V_\lambda(f)$ is finite.

\begin{theorem}\label{thm:Cpconvergencerate}
Let $f\in C^p([0,1])$ for $p\geq 1$ with $V_\lambda(f^{(p)})<\infty$ for some $0<\lambda\leq 1$.
Further assume that $\sup_{x\in[0,1]}|f^{(d)}(x)|\leq A$ for $1\leq d\leq \max(p-1,1)$. For random linear scrambling %with nonsingular $C$
\begin{align}\label{eq:Cpconvergencerate}
\Pr\Big(|\hat{\mu}_{\infty}-\mu|>\frac{C_1}{\sqrt{\eta}}2^{-(p+\lambda) (1-\epsilon) m}+\frac{C_2}{\sqrt{\eta}}2^{-2^{\frac{(1-\epsilon)m}{p}-1}}\Big)< \eta+C_3 2^{-\epsilon m}
\end{align}holds for any $\eta>0$ and $0\leq \epsilon<1$,
where
\begin{align*}
C_1&={\frac{2^{p+\lambda+2}}{\sqrt{p}(p-1)!}} V_\lambda(f^{(p)})\biggl(\sum_{N=p}^\infty N^{p-1}\Bigl(\frac{1}{2}\Bigr)^{\frac{N}{p+\lambda}}\biggr)^{\frac{1}{2}},\\
C_2&=4\sqrt{6}V_\lambda(f^{(p)})+8\sqrt{6}A,\quad\text{and}\\
C_3&=\frac{(p+\lambda)^p}{(p!)^2}\sum_{k=1}^{\infty} \frac{k^p}{2^k}+e-1.
\end{align*}
\end{theorem}
\begin{proof}
The proof resembles that of Theorem~\ref{thm:convergencerate} and is given in Appendix
\ref{app:proof:thm:Cpconvergencerate}.
\end{proof}

We have the following immediate consequences.

\begin{corollary}
Let $\med(\hat\mu_\infty)$ be the median of
the random variable $\hat\mu_\infty$.  Then under the conditions
of Theorem~\ref{thm:Cpconvergencerate}
$$|\med(\hat{\mu}_{\infty})-\mu|=o(2^{-(p+\lambda)m+\epsilon'm})$$ for any $\epsilon'>0$.
\end{corollary}
\begin{proof}
Apply Theorem~\ref{thm:Cpconvergencerate} with $\eta={1}/{m}$ and
$0<\epsilon<{\epsilon'}/{(p+\lambda)}$.
The probability in the right hand side of ~\eqref{eq:Cpconvergencerate} is below $1/2$ and the given bound for $|\hat\mu_\infty-\mu|$ is $o(2^{-(p+\lambda)m+\epsilon'm})$.
\end{proof}

Turning to the sample median, the smoother
$f$ is, the smaller are the probable
errors.  If we want to control the
expected squared error of the sample
median, then we will need $k$ to be large
enough and smoother $f$ will demand
larger $k$ as shown next. We do not
need $k$ to grow without bound.
\begin{corollary}\label{cor:medianofmeansholder}
Suppose that the scrambling has precision $E$ and $\hat{\mu}^{(k)}_{E}$ is the sample median of $2k-1$ independent copies of $\hat{\mu}_E$.
 Under the conditions of Theorem~\ref{thm:Cpconvergencerate},
\begin{compactenum}[(i)]
\item for any $\epsilon>0$, when $m$ is large enough so that $C_3 2^{-\epsilon m}<{1}/{16}$,
$$\Pr\Big(|\hat{\mu}^{(k)}_{E}-\mu|>4C_12^{-(p+\lambda) (1-\epsilon) m}+4C_22^{-2^{\frac{(1-\epsilon)m}{p}-1}}+\omega_f\Bigl(\frac{1}{2^E}\Bigr)\Big)< \Big(\frac{1}{2}\Big)^k$$
and
\item for any $\epsilon'>0$, if
$k=k(m)$ satisfies $\liminf_{m\to\infty}  \epsilon' k\geq 8(p+\lambda)^2$ and $E\geq (p+\lambda)m$, then
$$\e(|\hat{\mu}^{(k)}_{{E}}-\mu|^2)=o(2^{-2(p+\lambda)m+\epsilon'm}).$$
\end{compactenum}
\end{corollary}
\begin{proof}

By the argument in the proof of  Theorem~\ref{thm:samplemedian}, the probability in (i) is bounded by $4^k(\eta+C_3 2^{-\epsilon m})^k$
where $C_3$ is from Theorem~\ref{thm:Cpconvergencerate}.
Claim (i) follows once we choose $\eta={1}/{16}$
(making $\eta+C_32^{-\epsilon m}\le 1/8$)
and apply $|\hat{\mu}^{(k)}_{E}-\mu|\leq |\hat{\mu}^{(k)}_{\infty}-\mu|
+\omega_f(2^{-E})$.

To prove claim (ii), first notice that $f^{(1)}$ is either $\lambda$-Hölder continuous or differentiable, so it must be bounded over $[0,1]$ and  $\omega_f(2^{-E})\leq C2^{-E}$ for some constant $C$. Then $E\geq (p+\lambda)m$ implies
that $\omega_f(2^{-E})=O(2^{-(p+\lambda)m})$.
Now choose
$$\epsilon\in\Bigl(
\frac{\epsilon'}{4(p+\lambda)},
%<\epsilon<
\frac{\epsilon'}{2(p+\lambda)}\Bigr)$$ and $\eta=2^{-\epsilon m}$.
Similar to the way equation (\ref{eqn:MSE}) separates
contributions from large and small errors,
\begin{align*}
    \e(|\hat{\mu}_{E}-\mu|^2)&= O(2^{-2(p+\lambda)(1-\epsilon)m})+4^k(\eta+C_3 2^{-\epsilon m})^k\Vert f^{(1)}\Vert_\infty^2\\
    &=o(2^{-2(p+\lambda)m+\epsilon'm})+2^{-\epsilon m k+O(k)}.
\end{align*}
Since
$$\liminf_{m\to\infty}  \epsilon k\geq\liminf_{m\to\infty} \frac{\epsilon'k}{4(p+\lambda)}\geq 2 (p+\lambda),$$
we see that $2^{-\epsilon m k+O(k)}=o(2^{-2(p+\lambda)m+\epsilon'm})$. This proves the claim.
\end{proof}

We know that when $f\in C^p([0,1])$ and $f^{(p)}$ is $\lambda$-Hölder continuous, the optimal convergence rate is $O(n^{-p-\lambda})$.
See \cite{hein:nova:2002}.
Since a $\lambda$-Hölder continuous function has finite $V_\lambda(f)$,
the above corollary shows that if $k\to\infty$ as $m\to \infty$, then the sample median converges at the rate $o(n^{-p-\lambda+\epsilon})$ for any $\epsilon>0$, so it converges at almost the optimal rate.  The cost of computation
grows proportionally to $nk$.  Taking $k=\Omega(m)$ leads to a cost of $O(n\log(n))$.

\section*{Acknowledgments}

This work was supported by the U.S.\ National Science Foundation
under grant IIS-1837931.
We thank Mark Huber for a discussion about median of means,  Aleksei
Sorokin for a conversation about QMCPy and both Takashi Goda
and Pierre L'Ecuyer for discussions about super-polynomial convergence.
We thank Fred Hickernell
and the QMCPy team for making their
software available.  Finally, we are grateful
to the anonymous reviewers for their very helpful comments.

\bibliographystyle{apalike}
\bibliography{qmc}

%\begin{appendices}
\appendix
\section{Proof of Theorem~\ref{thm:errordecomposition}}
\label{app:proof:thm:errordecomposition}

Here we establish some lemmas and then prove Theorem~\ref{thm:errordecomposition}.
First, we introduce and slightly modify
the Walsh function notation for base $2$
from Appendix A of \cite{dick:pill:2010}
who credit \cite{pirs:1995}.
%For an integer base $b\ge2$ let $\omega_b=\exp(2\pi\sqrt{-1}/b)$. Our emphasis on $b=2$ means
%that we don't really need to use complex Walsh functions but
%we do so because the general notation is more familiar within QMC and we can directly quote results for that setting.
For $k\in\natu$,
we write $k=k_1+2k_2+\dots+2^{q-1} k_q$ with $k_q=1$.
Then the $k$'th Walsh function is
$$\walk(x)=(-1)^{k_1x_1+\dots+k_qx_q}=(-1)^{\vec{k}\cdot\vec{x}}$$
with both $\vec{x}$ and $\vec{k}$ taken to $q=q(k)$ bits.
The slight difference in our notation is that both of our
vectors $\vec{x}$ and $\vec{k}$ are indexed starting from $1$
while \cite{dick:pill:2010} index the bits of $k$ from zero and the bits of $x$ from $1$.
We also put $\mathrm{wal}_0(x)=1$ for all $x$.

Now for integers $k>0$, define $L_k=\{\ell\in [q(k)]\mid k_\ell=1 \}\in\lset$.
This $L_k$ will correspond to a set of row indices of $M$
when we interpret the binary expansion of $k$
as a 0--1 coding
for which rows are included.
Next, let $L_k(j)$ be the $j$-th largest element of $L_k$, $1\leq j\leq |L_k|$, and
for $1\le u\le |L_k|$ let $\Vert L_k\Vert_{1,u}$ be the sum of the largest $u$ elements of $L_k$, namely  $\Vert L_k \Vert_{1,u}=\sum_{j=1}^u L_k(j)$.
Finally, define $\chi_{r,k}$ to be
$$\chi_{r,k}=\int_{0}^1x^r{\text{wal}_{k}(x)}\rd x.$$
\begin{lemma}\label{lem:chiequalzero}
 If $0\leq r<|L_k|$, then $\chi_{r,k}=0$.
\end{lemma}
\begin{proof}
This follows from Lemma A.22 of \cite{dick:pill:2010}; see also Section 14.3 in that reference.
\end{proof}
\begin{lemma}\label{lem:chibound}
For any $u$ such that $1\leq  u\leq |L_k|$
$$|\chi_{r,k}|\leq  \frac{r!}{(r-u+1)!}\Big(\prod_{w=1}^{u}1+4^{-w+1}\Big)2^{-\Vert L_k \Vert_{1,u}-u}.$$
\end{lemma}
\begin{proof} The following proof uses the same proof strategy as that of Lemma 14.10 of \cite{dick:pill:2010}.

By equation (14.5) of \cite{dick:pill:2010} for $b=2$
\begin{equation}\label{eqn:recursion}
    \chi_{r,k}=-\frac{r}{2^{L_k(1)+1}}\Bigl(\chi_{r-1,k-2^{\kinfty}}-\sum_{c=1}^\infty \frac{1}{2^{c}}\chi_{r-1,k+2^{c+\kinfty}}\Bigr).
\end{equation}
When $u=1$, notice that for any $k$
$$|\chi_{r-1,k}|\leq \int_{0}^1 x^{r-1}\rd x=\frac{1}{r}.$$
Applying the above bound, we get
$$|\chi_{r,k}|\leq \frac{r}{2^{\kinfty+1}}\Bigl(\frac{1}{r}+\sum_{c=1}^\infty \frac{1}{2^{c}r}\Bigr)=\frac{2}{2^{\kinfty+1}}.$$
This proves the base case. Next we prove the bound for $u\geq 2$ by induction on~$u$.

If $|L_k|\geq u$, then $|L_{k-2^{\kinfty}}|=|L_k|-1 \geq u-1$ and $|L_{k+2^{c+\kinfty}}|+1=|L_k|+1\geq u-1$. Hence we can apply the induction hypothesis with $u-1$ to
equation \eqref{eqn:recursion} to get
\begin{align*}
   |\chi_{r,k}|&\leq \frac{r}{2^{\kinfty+1}}\bigg\{\frac{(r-1)!}{(r-u+1)!}\bigg(\prod_{w=1}^{u-1}1+4^{-w+1}\bigg)2^{-\Vert L_{k-2^{\kinfty}} \Vert_{1,u-1}-u+1} \\
   &\phe\,+\sum_{c=1}^\infty \frac{1}{2^{c}} \frac{(r-1)!}{(r-u+1)!}\bigg(\prod_{w=1}^{u-1}1+4^{-w+1}\bigg)2^{-\Vert L_{k+2^{c+\kinfty}} \Vert_{1,u-1}-u+1}\bigg\}\\
   &=\frac{r!}{(r-u+1)!2^u}\bigg(\prod_{w=1}^{u-1}1+4^{-w+1}\bigg) \\
   &\phe\times\Big\{2^{-\Vert L_{k-2^{\kinfty}} \Vert_{1,u-1}-\kinfty}+\sum_{c=1}^\infty \frac{1}{2^{c}}2^{-\Vert L_{k+2^{c+\kinfty}} \Vert_{1,u-1}-\kinfty}\Big\}.
\end{align*}
Now notice that
$$\kinfty-L_k(u)=\sum_{j=1}^{u-1} L_k(j)-L_k(j+1)\geq u-1$$
and then
\begin{align*}
  \Vert L_{k+2^{c+\kinfty}} \Vert_{1,u-1}+\kinfty
  &=\Vert L_{k} \Vert_{1,u}+c+2\kinfty-L_k(u-1)-L_k(u)\\
  &\geq \Vert L_{k} \Vert_{1,u}+c+2u-3.
\end{align*}
We also have $\Vert L_{k-2^{\kinfty}} \Vert_{1,u-1}+\kinfty=\Vert L_{k} \Vert_{1,u}$. Hence
\begin{align*}
   |\chi_{r,k}|&\leq \frac{r!}{(r-u+1)!2^u}\bigg(\prod_{w=1}^{u-1}1+4^{-w+1}\bigg)\Big\{2^{-\Vert L_{k} \Vert_{1,u}}+\sum_{c=1}^\infty \frac{1}{2^{c}}2^{-\Vert L_{k} \Vert_{1,u}-c-2u+3}\Big\}. \\
   &=\frac{r!}{(r-u+1)!2^{\Vert L_{k} \Vert_{1,u}+u}}\bigg(\prod_{w=1}^{u-1}1+4^{-w+1}\bigg)\Big\{1+\sum_{c=1}^\infty \frac{1}{4^{c}}2^{-2u+3}\Big\} \\
   &\leq \frac{r!}{(r-u+1)!}\bigg(\prod_{w=1}^{u}1+4^{-w+1}\bigg)2^{-\Vert L_k \Vert_{1,u}-u}.
\end{align*}
This completes the proof.
\end{proof}

\begin{lemma}\label{lem:Walshcoefficient}
Assume that $f$ is analytic on $[0,1]$ and $|f^{(d)}(1/2)|\leq A\alpha^dd!$ for some constants $A$ and $\alpha<2$ and all $d\ge1$. Then
$$|\hat{f}(k)|\leq 6A(|L_{k}|)!\Bigl(\frac{\alpha/2}{1-\alpha/2}\Bigr)^{|L_{k}|}2^{-\Vert L_k \Vert_{1}}.$$
\end{lemma}
\begin{proof}
First we split $\hat{f}(k)$ into two parts:
\begin{align}\label{eqn:fhat}
    |\hat{f}(k)|&=\Bigl|\int_{0}^1f(x){\text{wal}_{k}(x)}\rd x\Bigr|\notag\\
&\leq
\Bigl|\int_0^{1/2}f(x){\text{wal}_{k}(x)}\rd x\Bigr|+\Bigl|\int_{1/2}^1f(x){\text{wal}_{k}(x)}\rd x\Bigr|.
\end{align}
Let $y=2x-1$ and $g(y)=f(y/2+1/2)$. For $k=k_1+2k_2+\cdots+2^{q-1} k_q$, let $k'=(k-k_1)/2$. Then
\begin{align*}
    \int_{1/2}^1f(x){\text{wal}_{k}(x)}\rd x&=\int_{1/2}^1f(x)(-1)^{\sum_{\ell=1}^q k_\ell x_\ell}\rd x\\
    &=(-1)^{k_1}\int_{1/2}^1f(x)(-1)^{\sum_{\ell=2}^q k_\ell x_\ell}\rd x\\
    &=\frac{(-1)^{k_1}}{2}\int_{0}^1g(y){\text{wal}_{k'}(y)}dy\\
    &=\frac{(-1)^{k_1}}{2} \sum_{r=0}^\infty\frac{g^{(r)}(0)}{r!} \int_{0}^1y^r {\text{wal}_{k'}(y)}dy\\
    &=\frac{(-1)^{k_1}}{2} \sum_{r=0}^\infty\frac{g^{(r)}(0)}{r!} \chi_{r,k'}.
\end{align*}
Now by Lemma~\ref{lem:chiequalzero}, $\chi_{r,k'}=0$ if $r< |L_{k'}|$.
By Lemma~\ref{lem:chibound} with $u=|L_{k'}|$ and $|g^{(r)}(0)|=|f^{(r)}(1/2)|/2^r\leq Ar!(\alpha/2)^r$,
\begin{align*}
   \Bigl|\frac{(-1)^{k_1}}{2}\sum_{r=0}^\infty\frac{g^{(r)}(0)}{r!} \chi_{r,k'}\Bigr|
   &\leq \sum_{r=|L_{k'}|}^\infty \frac{Ar!(\alpha/2)^r}{(r-|L_{k'}|+1)!}\Big(\prod_{w=1}^{|L_{k'}|}1+4^{-w+1}\Big)2^{-\Vert L_{k'} \Vert_{1}-|L_{k'}|-1} \\
   &\leq 3A2^{-\Vert L_{k'} \Vert_{1}-|L_{k'}|-1}\sum_{r=|L_{k'}|}^\infty \frac{r!(\alpha/2)^r}{(r-|L_{k'}|+1)!},
\end{align*}
where we have used 
$$\prod_{w=1}^{|L_{k'}|}1+4^{-w+1}< 2\exp\biggl(\,\sum_{w=2}^\infty 4^{-w}\biggr)<3.$$
Now 
\begin{align*}
\sum_{r=|L_{k'}|}^\infty
\frac{r!(\alpha/2)^r}{(r-|L_{k'}|+1)!}
&=(|L_{k'}|-1)!\sum_{r=|L_{k'}|}^\infty
{r \choose r-|L_{k'}|+1}(\alpha/2)^r\\
&=(|L_{k'}|-1)!(\alpha/2)^{|L_{k'}|-1}\sum_{r=1}^\infty
{r+|L_{k'}|-1 \choose r}(\alpha/2)^r\\
&\overset{(i)}{=}(|L_{k'}|-1)!(\alpha/2)^{|L_{k'}|-1}\Bigl(\frac{1}{(1-\alpha/2)^{|L_{k'}|}}-1\Bigr)\\
&=(|L_{k'}|)!\Bigl(\frac{\alpha/2}{1-\alpha/2}\Bigr)^{|L_{k'}|}\frac{1-(1-\alpha/2)^{|L_{k'}|}}{|L_{k'}|\alpha/2}\\
&\overset{(ii)}{\leq} (|L_{k'}|)!\Bigl(\frac{\alpha/2}{1-\alpha/2}\Bigr)^{|L_{k'}|}
\end{align*}
where $(i)$ uses the Taylor expansion of $(1-x)^{-n}$ around $x=0$ evaluated at $x=\alpha/2$ and $n=|L_{k'}|$ and $(ii)$ uses $1-(1-x)^{n}\leq nx$ for $x\geq 0$. Therefore
\begin{align*}
   \Bigl|\frac{(-1)^{k_1}}{2}\sum_{r=0}^\infty\frac{g^{(r)}(0)}{r!} \chi_{r,k'}\Bigr|
   \leq 3A(|L_{k'}|)!\Bigl(\frac{\alpha/2}{1-\alpha/2}\Bigr)^{|L_{k'}|}2^{-\Vert L_{k'} \Vert_{1}-|L_{k'}|-1}.
\end{align*}

From the definition of $L_k$, it is easy to see
$|L_{k}|-1\leq |L_{k'}|\leq |L_{k}|$ and
$$\Vert L_{k'} \Vert_{1}=\sum_{\ell\in L_{k'}}\ell= \sum_{\ell\in L_{k}}(\ell-1)=\Lknorm-|L_k|,$$
from which we get
$$\Bigl|\int_{1/2}^1f(x){\text{wal}_{k}(x)}\rd x\Bigr|
\leq 3A(|L_{k}|)!\Bigl(\frac{\alpha/2}{1-\alpha/2}\Bigr)^{|L_{k}|}2^{-\Vert L_k \Vert_{1}}.$$
We can bound $|\int_0^{1/2}f(x){\text{wal}_{k}(x)}\rd x|$ in a similar way.
Now the conclusion follows from equation \eqref{eqn:fhat}.
\end{proof}

\subsubsection*{Proof of Theorem~\ref{thm:errordecomposition}}

By the Walsh function decomposition,
$f(x)=\sum_{k=0}^\infty \hat{f}(k) \text{wal}_{k}(x)$
and so
$$\hat{\mu}_\infty-\mu=\sum_{k=1}^\infty \hat{f}(k) \frac{1}{2^m}\sum_{i=0}^{2^m-1} \text{wal}_{k}(x_i).$$
Now by equation (\ref{eqn:xequalMCiplusD}),
$$\text{wal}_{k}(x_i)=(-1)^{\sum_{\ell\in L_k}x_{i,\ell}}=(-1)^{ \sum_{\ell\in L_k}M(\ell,:) \vec{a}_i +\sum_{\ell\in L_k}D_\ell}.$$
If $\sum_{\ell\in L_k}M(\ell,:)=\bszero$, then
$$\frac{1}{2^m}\sum_{i=0}^{2^m-1} \text{wal}_{k}(x_i)=\frac{1}{2^m}\sum_{i=0}^{2^m-1}(-1)^{\sum_{\ell\in L_k}D_\ell}=(-1)^{\norm{D}}.$$
Otherwise at least one entry of $\sum_{\ell\in L_k}M(\ell,:)$ is nonzero. Because $C$ is nonsingular and $\vec{a}_i=C\vec{i}$,  $\{\vec{a}_i\mid 0\leq i<2^m\}=\{0,1\}^m$, and so
\begin{align*}
    \frac{1}{2^m}\sum_{i=0}^{2^m-1} \text{wal}_{k}(x_i)&=(-1)^{\norm{D}}\frac{1}{2^m}\sum_{i=0}^{2^m-1}(-1)^{\sum_{\ell\in L_k}M(\ell,:)\vec{a}_i }\\
    &=(-1)^{\norm{D}}\prod_{q=1}^m (1+(-1)^{\sum_{\ell\in L_k}M(\ell,q)})\\
    &=0.
\end{align*}
Therefore,
\begin{equation}\label{eqn:proofoftheorem1}
       \sum_{k=1}^{\infty} \hat{f}(k) \frac{1}{2^m}\sum_{i=0}^{2^m-1} \text{wal}_{k}(x_i)=\sum_{k=1}^{\infty}\bsone\biggl\{\sum_{\ell\in L_k}M(\ell,:)=\bszero\biggr\} \hat{f}(k) (-1)^{\norm{D}}.
\end{equation}
The conclusion follows once we define $B_{L_k}=\hat{f}(k)2^{\Lknorm}$ and
notice that $\{L_k\mid k\geq 1\}=\lset$. The bound on $B_{L_k}$ follows directly from Lemma~\ref{lem:Walshcoefficient}. $\Box$

\section{Proof of Lemma~\ref{lem:combinatorics}}\label{app:proof:lem:combinatorics}
Let $q(N)=|\{L\in \lset \mid \Lnorm = N\}|$. Each $L$ with $\Lnorm = N$ corresponds to a partition of $N$ into distinct positive integers. It is known from combinatorics that
\begin{align}\label{eq:fromcomb}
q(N)\sim\frac{1}{4\times3^{1/4}N^{3/4}}\exp\Bigl(\pi\sqrt{\frac{N}{3}}\Bigr)
\end{align}
where $\sim$ means asymptotically equivalent (the ratio of the two sides converges to $1$ as $N\to\infty$).
See note VII.24 in \cite{flaj:sedg:2009}. Because each $q(N)$ is positive and the sum $\sum_{n=1}^N q(n)$ grows to infinity as $N\to\infty$, 
sums of the first $N$ members of
the left hand side of~\eqref{eq:fromcomb}
are asymptotic to the corresponding
sums of the right hand side:
$$\bigl|\{L\in \lset \mid \Lnorm \leq N\}\bigr|=\sum_{n=1}^N q(n)\sim \sum_{n=1}^N \frac{1}{4\times 3^{1/4}n^{3/4}}\exp\Bigl(\pi\sqrt{\frac{n}{3}}\Bigr).$$
Because the function ${x^{-3/4}}\exp(\pi\sqrt{{x}/{3}})$ has positive derivative for $x> 3^{3/2}/(2\pi)\approx 0.82$, we have
\begin{align}\label{eqn:sandwich}
\int_1^{N} \frac{1}{x^{3/4}}\exp\Bigl(\pi\sqrt{\frac{x}{3}}\Bigr)\rd x
&\le \sum_{n=1}^N \frac{1}{n^{3/4}}\exp\Bigl(\pi\sqrt{\frac{n}{3}}\Bigr)\notag\\
&\le \int_1^{N+1} \frac{1}{x^{3/4}}\exp\Bigl(\pi\sqrt{\frac{x}{3}}\Bigr)\rd x
\end{align}
where the first inequality follows
from integrating $x^{-3/4}\exp(\pi\sqrt{x/4})\bsone\{x\ge1\}$
over $[0,N]$.
By the change of variable $y=(\pi^2 {x}/{3})^{1/4}$, we get
$$\int_1^{N} \frac{1}{x^{3/4}}\exp\Bigl(\pi\sqrt{\frac{x}{3}}\Bigr)\rd x
=4\Bigl(\frac{3}{\pi^2}\Bigr)^{1/4}\int_{0}^{(\pi^2 {N}/{3})^{1/4}}e^{y^2}dy+O(1).$$

We will use the Dawson function
\begin{align}\label{eq:defdaw}
\daw(z) = e^{-z^2}\int_0^ze^{y^2}\rd y
\end{align}
for $z\ge0$.
For large $z$ there is an asymptotic expansion
$$
\daw(z)\sim \frac1{2z}+\frac1{2^2z^3}+\frac{1\cdot 3}{2^3z^5}+\frac{1\cdot3\cdot5}{2^4z^7} + \cdots
$$
from~\cite{leth:wens:1991}.
The first order approximation $\daw(z)\sim 1/(2z)$ is enough
for our purposes. Using~\eqref{eq:defdaw} we can write
\begin{align*}\int_{0}^{(\pi^2 \frac{N}{3})^{1/4}}e^{y^2}dy&=\exp\Bigl(\pi \sqrt{\frac{N}{3}}\Bigr)
\daw\Bigl(\Bigl(\pi^2 \frac{N}{3}\Bigr)^{1/4}\Bigr)\\
&\sim
\frac{1}{2(\pi^2 {N}/{3})^{1/4}}\exp\Bigl(\pi \sqrt{\frac{N}{3}}\Bigr).
\end{align*}
So we conclude that
$$\int_1^{N} \frac{1}{x^{3/4}}\exp\Bigl(\pi\sqrt{\frac{x}{3}}\Bigr)\rd x\sim \frac{2\times{3^{1/2}}}{\pi N^{1/4}}\exp\Bigl(\pi \sqrt{\frac{N}{3}}\Bigr).$$
From equation (\ref{eqn:sandwich}), the sum has the same asymptotic growth rate. Hence
$$\bigl|\{L\in \lset \mid \Lnorm \leq N\}\bigr|
\sim \sum_{n=1}^N \frac{1}{4\times3^{1/4}n^{3/4}}\exp\Bigl(\pi\sqrt{\frac{n}{3}}\Bigr)\sim \frac{3^{1/4}}{2\pi N^{1/4}}\exp\Bigl(\pi\sqrt{\frac{N}{3}}\Bigr).$$
In other words,
$$\lim_{N\to\infty} N^{1/4}\exp\Bigl(-\pi\sqrt{\frac{N}{3}}\Bigr) \bigl|\{L\in \lset \mid \Lnorm \leq N\}\bigr|=\frac{3^{1/4}}{2\pi}.$$
Now define $\lambda={3(\log(2))^2}/{\pi^2}$. Let $N=\lambda m^2$ in the above equation and notice that $\pi\sqrt{N/3}={ m\log(2)}$,
$$\lim_{m\to\infty} \lambda^{1/4}\sqrt{m}2^{-m} \bigl|\{L\in \lset \mid \Lnorm \leq \lambda m^{2}\}\bigr|=\frac{3^{1/4}}{2\pi}.$$

A numerical calculation summing the PartitionsQ function of Mathematica from $1$ to $N$ shows that
$$\bigl|\{L\in \lset \mid \Lnorm \leq N\}\bigr|< \frac{0.242}{N^{1/4}}\exp\Bigl(\pi\sqrt{\frac{N}{3}}\Bigr)$$
for $1\leq N\leq 40{,}000$. Because $\lambda (512)^2\approx 38{,}284$, the above inequality shows that
\begin{align*}
\bigl|\{L\in \lset \mid \Lnorm \leq \lambda m^2\}\bigr|< \frac{0.242 \times 2^{m}}{\lambda^{\frac{1}{4}}\sqrt{m}}< \frac{0.4 \times 2^{m}}{\sqrt{m}}
\end{align*}
for $1\leq m\leq 512$. This completes the proof.
$\qed$

\section{Proof of Theorem~\ref{thm:Cpconvergencerate}}
\label{app:proof:thm:Cpconvergencerate}

Our proof strategy is essentially the same as that of Theorem~\ref{thm:convergencerate}. First we establish the counterpart of Theorem~\ref{thm:errordecomposition} when $f$ is only finitely differentiable. Recall that in Appendix A, we defined $L(j)$ to be the $j$th-largest element of $L$ and $\Vert L\Vert_{1,u}$ to be the sum of the largest $u$ elements of $L$.
\begin{lemma}\label{lem:Cperror}
Assume that $f\in C^p([0,1])$ for $p\geq 1$ and $V_\lambda(f^{(p)})<\infty$ for some $0<\lambda\leq 1$. Further assume that $\sup_{x\in[0,1]}|f^{(d)}(x)|\leq A$ for $1\leq d\leq \max(p-1,1)$. Then
\begin{align*}
\hat{\mu}_{\infty}-\mu&=\sum_{L\in \lset}\bsone\{|L|\leq p,\sum_{\ell\in L}M(\ell,:)=\bszero\}  B_{L} (-1)^{\norm{D(L)}} 2^{-\Lnorm}\\
&+\sum_{L\in \lset}\bsone\{|L|>p,\sum_{\ell\in L}M(\ell,:)=\bszero\}  B_{L} (-1)^{\norm{D(L)}} 2^{-\Vert L\Vert_{1,p}-\lambda L(p+1)},
\end{align*}
where $B_L$ is a coefficient depending on $L$ that satisfies
\begin{equation}\label{eqn:CpBLbounds}
    |B_L|\leq
\begin{cases}
    4V_\lambda(f^{(p)})+8A,&  \text{ if } |L|\leq p\\
    4 V_\lambda(f^{(p)}),&\text{ if } |L|>p.
    \end{cases}
\end{equation}
% \begin{equation}\label{eqn:CpBLbound2}
%     |B_L|\leq 8 V_\lambda(f^{(p)})\text{ if } |L|>p
% \end{equation}
\end{lemma}
\begin{proof}
If $|L_k|>p$, then by the first part of
  \citet[Theorem 14.15]{dick:pill:2010}, with $b=2$,
$$|\hat{f}(k)|\leq 4 V_\lambda(f^{(p)})2^{-\Vert L\Vert_{1,p}-\lambda L(p+1)}.$$

For $|L_k|=p$ and $|L_k|>p$, the second and third parts, respectively, of   \citet[Theorem 14.15]{dick:pill:2010} apply. We will use the larger upper bound from the third part. Because we assume that $\sup_{x\in[0,1]}|f^{(r)}(x)|\leq A$ for $1\leq r\leq \max(p-1,1)$,
\begin{align*}
    |\hat{f}(k)|&\leq 2^{-\Lknorm}\Big(4 V_\lambda(f^{(p)})2^{-\lambda L_k(p)}+3\sum_{r=|L_k|}^{p-1}\frac{|f^{(r)}(0)|}{(r-|L_k|+1)!}+\Bigl|\int_0^1 f^{(|L_k|)}(x)\rd x\Bigr|\Big)\\
    &\leq 2^{-\Lknorm}\Big(4 V_\lambda(f^{(p)})+3A\sum_{r=1}^{\infty}\frac{1}{r!}\\
    &\phe\ +\min\Bigl(\sup_{x\in[0,1]}|f^{(|L_k|)}(x)|,|f^{(|L_k|-1)}(1)-f^{(|L_k|-1)}(0)|\Bigr)\Big)\\
    &< (4V_\lambda(f^{(p)})+8A) 2^{-\Lknorm}
\end{align*}
because $3(e-1)+{2}<8$. We add $2$ instead of $1$ here to handle the case $|L_k|=p$, where $|f^{(p-1)}(x)|<A$ but $|f^{(p)}(x)|$ might not be smaller than $A$.

The conclusion follows once we use this estimate for $\hat f(k)$ in equation (\ref{eqn:proofoftheorem1}) and define $B_{L_k}=\hat{f}(k)2^{\Lknorm}$ if $|L_k|\leq p$ and $B_{L_k}=\hat{f}(k)2^{\Vert L_k\Vert_{1,p}+\lambda L_k(p+1)}$ if $|L_k|>p$.
\end{proof}

Next we prove the counterpart of Lemma~\ref{lem:concentration}. To shorten some lengthy expressions we use the notation \begin{align}\label{eq:lsetm}
\lset(M)=\Bigl\{L\in\lset\bigm| \sum_{\ell\in L}M(\ell,:)=\bszero\Bigr\}.
\end{align}
\begin{lemma}\label{lem:Cprelaxedver}
In random linear scrambling, when $m\geq 1$, for any $0\leq \epsilon<1$
\begin{align}\label{eq:boundp1}
&\Pr\Bigl(\min\bigl\{\Lnorm\mid L\in \lset(M), |L|\leq p %,\sum_{\ell \in L} %M(\ell,:)=\bszero\bigr
\}\leq 2^{\frac{(1-\epsilon)m}{p}}\Bigr)
< (e-1) 2^{-\epsilon m},
\end{align}
and
\begin{multline}\label{eq:boundp2}
\Pr\Bigl(\min\bigl\{\Vert L\Vert_{1,p}+\lambda L(p+1)\bigm| L\in \lset(M), |L|> p
%,\sum_{\ell \in L} M(\ell,:)=\bszero
\bigr\}
\leq (1-\epsilon)(p+\lambda)m\Bigr)
< C_{p,\lambda} 2^{-\epsilon m}
\end{multline}
where $C_{p,\lambda}$ is given by equation (\ref{eqn:Cplambda}).
\end{lemma}
\begin{proof}
We will make frequent use of this quantity:
$$q(N,d)=\bigl|\bigl\{L\in \lset \bigm| |L|=d, \Lnorm = N \bigr\}\bigr|.$$
Now let
$$
\begin{pmatrix}
K(1)\\
K(2)\\
\vdots\\
K(d-1)\\
K(d)
\end{pmatrix}
=\begin{pmatrix}
L(1)\\
L(2)\\
\vdots\\
L(d-1)\\
L(d)
\end{pmatrix}
-\begin{pmatrix}
d-1\\
d-2\\
\vdots\\
1\\
0
\end{pmatrix}.
$$
This provides a bijection between strictly decreasing positive integers $L(j)$ that sum to $N$ and non-increasing positive integers $K(j)$ that sum to $N-d(d-1)/2$.
It follows that $q(N,d)$ equals  the number of ways to partition $N-{d(d-1)}/{2}$ into $d$ positive integers. Hence by \citet[4.2.6]{andrews1984theory}
\begin{equation}\label{eqn:qNd}
    q(N,d)\leq
    p_d(N-d(d-1)/2)=
    \frac{1}{d!}{N-1\choose d-1},
\end{equation}
where $p_M(n)$ is Andrews' notation for the number of partitions of $n$ into $M$ positive integers. Therefore
\begin{align}\label{eq:boundqsum}
\begin{split}
|\{L\in \lset \mid |L|=d, \Lnorm \leq N \}|&=\sum_{n=1}^N q(n,d)\leq \sum_{n=1}^N \frac{1}{d!}{n-1\choose d-1}\\&=\frac{1}{d!}{N\choose d}\le \frac{N^d}{(d!)^2}
\end{split}
\end{align}
where we have used the `upper summation' identity for binomial coefficients
\citep{grahamconcrete} to simplify the sum over $n$.

Now we take the union bound and use  $\Pr(\sum_{\ell \in L} M(\ell,:)=\bszero)\leq 2^{-m}$, to get
\begin{align*}
    &\Pr\Bigl(\min\bigl\{\Lnorm\mid L\in \lset(M), |L|\leq p
    %, \sum_{\ell \in L} M(\ell,:)=\bszero\bigr
    \}\leq 2^{\frac{(1-\epsilon)m}{p}}\Bigr)\\
    &\leq \frac{1}{2^m} \sum_{d=1}^p\frac{2^{\frac{d(1-\epsilon)m}{p}}}{(d!)^2}
        \leq \frac{1}{2^m} \sum_{d=1}^p\frac{2^{m(1-\epsilon)}}{(d!)^2}
        \leq \frac{1}{2^{\epsilon m}} \sum_{d=1}^\infty \frac{1}{d!}=(e-1) 2^{-\epsilon m},
\end{align*}
establishing~\eqref{eq:boundp1}.

For~\eqref{eq:boundp2} we must count the sets $L$ with $\Vert L\Vert_{1,p}+\lambda L(p+1)\leq N$. We make a separate count for each value of $L(p+1)$. Let $L(p+1)=k$. Similar to the bijection above, we can make a one to one correspondence between $L(1),\dots,L(p)$ that are strictly decreasing, larger than  $L(p+1)=k$ and sum to $\Vert L\Vert_{1,p}\le N-\lceil \lambda k\rceil$ and integers $L(1)-k>L(2)-k>\dots>L(p)-k>0$
with a sum at most $N-\lceil\lambda p\rceil-kp$. The smallest relevant $|L|$ is $p+1$. For that we get
\begin{align}\label{eq:boundlambdasum}
    &\bigl|\{L\in \lset \bigm| |L|=p+1, L(p+1)=k, \Vert L\Vert_{1,p}+\lambda L(p+1) \leq N\}\bigr|\nonumber\\
    &\leq \sum_{n=1}^N q(n-kp-\lceil\lambda k\rceil,p)=\sum_{n=1}^{N-kp-\lceil\lambda k\rceil} q(n,p)\leq \frac{(N-kp-\lceil\lambda k\rceil)^p}{(p!)^2}
\end{align}
by the bound from~\eqref{eq:boundqsum}.
Now if we allow $|L|>p$, the largest $|L|$ could be is $p+k-1$ because $L(p+1)=k$ and $L(j)$ are strictly decreasing. Each of the distinct values $k-1$ through $1$ could either appear or not appear among the $L(j)$ for $j>p+1$. Those that do appear must do so in strictly decreasing order. As a result, including cases with $|L|>p$ raises the count by a factor of $2^{k-1}$.
Summing over $k$ we get
\begin{align*}
    &\phe\ \, \bigl|\bigl\{L\in \lset \bigm| |L|>p, \Vert L\Vert_{1,p}+\lambda L(p+1) \leq N\bigr\}\bigr|\\
    &\leq \sum_{k=1}^{N} \one\{kp+\lceil\lambda k\rceil\leq N\} \frac{(N-kp-\lceil\lambda k\rceil)^p}{(p!)^2} 2^{k-1}.
\end{align*}
Letting $K^*=\lfloor {N}/({p+\lambda})\rfloor$ and $k'=K^*-k+1$,
\begin{align*}
   &\phe\ \,\sum_{k=1}^{N} \one\{kp+\lceil\lambda k\rceil\leq N\} (N-kp-\lceil\lambda k\rceil)^p 2^{k-1} \\
   & \leq \sum_{k'=1}^{K^*} (N-(p+\lambda)(K^*-k'+1))^p 2^{K^*-k'}\\
   & \leq 2^{K^*}(p+\lambda)^{p} \sum_{k'=1}^{\infty} k'^p 2^{-k'}.
\end{align*}
The last sum is clearly convergent. Therefore
$$|\{L\in \lset \mid |L|>p, \Vert L\Vert_{1,p}+\lambda L(p+1) \leq N\}|\leq C_{p,\lambda} 2^{\frac{N}{p+\lambda}}$$
where
\begin{equation}\label{eqn:Cplambda}
    C_{p,\lambda}=\frac{(p+\lambda)^p}{(p!)^2}\sum_{k=1}^{\infty} \frac{k^p}{2^k}.
\end{equation}
Applying the union bound and using $\Pr(\sum_{\ell \in L} M(\ell,:)=\bszero)\leq 2^{-m}$,
\begin{align*}
    &\Pr\Bigl(\min\bigl\{\Vert L\Vert_{1,p}+\lambda L(p+1)\bigm| L\in \lset(M), |L|> p
    %, \sum_{\ell \in L} M(\ell,:)=\bszero
    \bigr\}\leq (1-\epsilon)(p+\lambda)m\Bigr)\\
&\leq \frac{1}{2^m} C_{p,\lambda} 2^{\frac{(1-\epsilon)(p+\lambda)m}{p+\lambda}}=C_{p,\lambda}2^{-\epsilon m}.\qedhere\end{align*}
\end{proof}

\subsection*{Proof of Theorem~\ref{thm:Cpconvergencerate}}
Using $\lset(M)$ from~\eqref{eq:lsetm}
define the event \begin{align*}
H&=\Big\{\min\{\Lnorm\mid L\in \lset(M), |L|\leq p
%, \sum_{\ell \in L} M(\ell,:)=\bszero
\}> 2^{\frac{(1-\epsilon)m}{p}}\Big\}\\
&\phe\,\bigcap \Big\{\min\{\Vert L\Vert_{1,p}+\lambda L(p+1)\mid L\in \lset(M), |L|> p
%, \sum_{\ell \in L} M(\ell,:)=\bszero\}
> (1-\epsilon)(p+\lambda)m\Big\}.
\end{align*}
Lemma~\ref{lem:Cprelaxedver} shows that $\Pr(H^c)<(C_{p,\lambda}+e-1)2^{-\epsilon m}$.

Now as in equations (\ref{eqn:sumoverN}) and (\ref{eqn:PH}),
\begin{multline*}
    \e\bigl(\var(\hat{\mu}_{\infty}-\mu\giv M)\mid H\bigr)\leq \frac{2^{-m}}{\Pr(H)}\biggl(\,\sum_{L\in \lset,|L|\leq p}\one\Bigl\{\Lnorm>2^{\frac{(1-\epsilon)m}{p}}\Bigr\}B^2_{L}  4^{-\Lnorm}\\
    +\sum_{L\in \lset,|L|> p}\one\Bigl\{\frac{\Vert L\Vert_{1,p}+\lambda L(p+1)}{p+\lambda}>(1-\epsilon)m\Bigr\}B^2_{L}  4^{-\Vert L\Vert_{1,p}-\lambda L(p+1)} \biggr).
\end{multline*}
Lemma~\ref{lem:Cperror}
provides two uniform bounds on $B_L^2$
depending on whether $|L|\le p$ or $|L|>p$.
We will bound the sum above exclusive
of the $B_L^2$ factors for now and
then multiply in those factors below.
We can easily bound the first sum above by removing the restriction $|L|\leq p$ and using equation (\ref{eqn:Bidar}) to bound the number of $L\in\lset$ with $|L|=N$. This yields
\begin{align*}
    \sum_{L\in \lset,|L|\leq p}\one\{\Lnorm>2^{\frac{(1-\epsilon)m}{p}}\} 4^{-\Lnorm}&=\sum_{N=\bigl\lceil 2^{\frac{(1-\epsilon)m}{p}}\bigr\rceil}^\infty \frac{1}{4^N} |\{L\in \lset\mid |L|\leq p, \Lnorm=N\}|\\
    &\leq \sum_{N=\bigl\lceil 2^{\frac{(1-\epsilon)m}{p}}\bigr\rceil}^\infty \frac{1}{4^N} \frac{\pi \exp\Bigl(\pi\sqrt{\frac{N}{
    3}}\Bigr)}{2\sqrt{3}}\\
    &= \frac{\pi}{2\sqrt{3}} \sum_{N=\bigl\lceil 2^{\frac{(1-\epsilon)m}{p}}\bigr\rceil}^\infty \exp\Bigl(\pi \sqrt{\frac{N}{3}}-\log(4)N\Bigr)\\
    &\overset{(\mathrm{i})}{\leq } \frac{\pi}{2\sqrt{3}} \sum_{N=\bigl\lceil 2^{\frac{(1-\epsilon)m}{p}}\bigr\rceil}^\infty \exp\Bigl(\frac{\pi^2}{12\log(2)}+\log(2)N-\log(4)N\Bigr)\\
    &\leq \frac{\pi}{2\sqrt{3}} \exp\Bigl(\frac{\pi^2}{12\log(2)}\Bigr) \sum_{N=\bigl\lceil 2^{\frac{(1-\epsilon)m}{p}}\bigr\rceil}^\infty \frac{1}{2^N}\\
    &=\frac{\pi}{\sqrt{3}} \exp\Bigl(\frac{\pi^2}{12\log(2)}\Bigr) 2^{-\bigl\lceil 2^{\frac{(1-\epsilon)m}{p}}\bigr\rceil}\\
    &\leq 6 \times 2^{-2^{\frac{(1-\epsilon)m}{p}}}
\end{align*}
where (i) follows from the inequality $a\leq a^2/c+c/4$ with $a=\pi \sqrt{N/3}$ and $c=\pi^2/(3\log(2))$.

To bound the second sum, we consider separate cases for each value of $L(p+1)$ and $\Vert L\Vert_{1,p}$. The argument is similar to the one used in Lemma~\ref{lem:Cprelaxedver}, but not so similar that we could just cite that lemma. First
\begin{multline*}
\sum_{L\in \lset,|L|> p}\one\biggl\{\frac{\Vert L\Vert_{1,p}+\lambda L(p+1)}{p+\lambda}>(1-\epsilon)m\biggr\}  4^{-\Vert L\Vert_{1,p}-\lambda L(p+1)}\\
    =\sum_{k=1}^\infty  \sum_{N=\lceil(p+\lambda)(1-\epsilon)m-\lambda k \rceil}^\infty \frac{1}{4^{N+\lambda k}}\bigl|\bigl\{L\in \lset \bigm| |L|>p, L(p+1)=k, \Vert L\Vert_{1,p}= N\bigr\}\bigr|.
\end{multline*}
The bijection used to derive equation~\eqref{eq:boundlambdasum} shows that
\begin{multline*}\bigl|\{L\in \lset \bigm| |L|=p+1, L(p+1)=k, \Vert L\Vert_{1,p}= N\}\bigr|\\=\bigl|\{L\in \lset \bigm| |L|=p, \Vert L\Vert_{1,p}= N-kp\}\bigr|,
\end{multline*}%we can subtract $k$ from each $L(j), 1\leq j\leq p$, which generates a set of $k$ distinct positive integers with sum $N-kp$.
and so by equation (\ref{eqn:qNd})
\begin{align*}
&\phe\, \bigl|\{L\in \lset \bigm| |L|=p+1, L(p+1)=k, \Vert L\Vert_{1,p}= N\}\bigr|\\&=q(N-kp,p)\leq \frac{1}{p!}{N-kp-1\choose p-1}.\end{align*}
The same argument used below equation~\eqref{eq:boundlambdasum} shows that allowing $|L|>p$ raises the count by a factor of $2^{k-1}$.
Hence
\begin{align*}
    &\sum_{k=1}^\infty  \sum_{N=\lceil(p+\lambda)(1-\epsilon)m-\lambda k \rceil}^\infty \frac{1}{4^{N+\lambda k}}|\{L\in \lset \mid |L|>p, L(p+1)=k, \Vert L\Vert_{1,p}= N\}|\\
    &\leq \sum_{k=1}^\infty  \sum_{N=\lceil(p+\lambda)(1-\epsilon)m-\lambda k \rceil}^\infty \frac{1}{4^{N+\lambda k}}\frac{1}{p!}{N-kp-1\choose p-1}2^{k-1}\\
    &\overset{(\mathrm{i})}{\leq} \frac{1}{p!}\sum_{k=1}^\infty \frac{2^{k-1}}{4^{(p+\lambda) k}} \sum_{N=\max(\lceil(p+\lambda)(1-\epsilon)m-\lambda k \rceil,(k+1)p)}^\infty \frac{1}{4^{N-kp}}\frac{(N-kp)^{p-1}}{(p-1)!}\\
   & \overset{(\mathrm{ii})}{\leq} \frac{1}{p!(p-1)!}\sum_{k=1}^\infty \frac{2^{k-1}}{4^{(p+\lambda) k}}\sum_{N=\max(\lceil(p+\lambda)((1-\epsilon)m- k) \rceil,p)}^\infty \frac{N^{p-1}}{4^N}
\end{align*}
where (i) uses ${N-kp-1\choose p-1}=0$ if $N<(k+1)p$ and ${N-kp-1\choose p-1}<(N-kp)^{p-1}/(p-1)!$ if $N\geq (k+1)p$ and (ii) shifts the index $N$ by $kp$.
Letting $k'=\lfloor m(1-\epsilon)\rfloor-k$, we see that
\begin{align*}
& \sum_{k=1}^\infty \frac{2^{k-1}}{4^{(p+\lambda) k}}\sum_{N=\max(\lceil(p+\lambda)((1-\epsilon)m- k) \rceil,p)}^\infty \frac{N^{p-1}}{4^N}\\
&\le\frac{2^{\lfloor m(1-\epsilon)\rfloor-1}}{4^{(p+\lambda) \lfloor m(1-\epsilon)\rfloor}}\sum_{k'=-\infty}^{\lfloor m(1-\epsilon)\rfloor-1} \frac{2^{-k'}}{4^{-(p+\lambda)k'}}\sum_{N=\max(\lceil k'(p+\lambda) \rceil,p)}^\infty \frac{N^{p-1}}{4^N}\\
    &\leq \frac{2^{\lfloor m(1-\epsilon)\rfloor-1}}{4^{(p+\lambda) \lfloor m(1-\epsilon)\rfloor}}\sum_{k'=-\infty}^\infty \frac{2^{-k'}}{4^{-(p+\lambda)k'}}\sum_{N=\max(\lceil k'(p+\lambda) \rceil,p)}^\infty \frac{N^{p-1}}{4^N}\\
    &=\frac{2^{\lfloor m(1-\epsilon)\rfloor-1}}{4^{(p+\lambda) \lfloor m(1-\epsilon)\rfloor}}\sum_{N=p}^\infty \frac{N^{p-1}}{4^N} \sum_{k'=-\infty}^{\lfloor \frac{N}{p+\lambda}\rfloor}2^{(2p+2\lambda-1)k'}\\
    &= \frac{2^{\lfloor m(1-\epsilon)\rfloor-1}}{4^{(p+\lambda) \lfloor m(1-\epsilon)\rfloor}}\sum_{N=p}^\infty \frac{N^{p-1}}{4^N}2^{(2p+2\lambda-1)\lfloor\frac{N}{p+\lambda}\rfloor} \frac{2^{2p+2\lambda-1}}{2^{2p+2\lambda-1}-1}\\
     &\leq \frac{2^{ m(1-\epsilon)-1}}{4^{(p+\lambda)  (m(1-\epsilon)-1)}} \sum_{N=p}^\infty N^{p-1}(2^{-\frac{1}{p+\lambda}})^N \times 2
\end{align*}
where the last inequality uses $2^{2p+2\lambda-1}\geq 2$ because $p\geq 1$ and $\lambda>0$.

%%For the sum over $k\geq \lceil m(1-\epsilon)\rceil $, because we assume $p\geq 1$
%%$$\sum_{k=\lceil(1-\epsilon)m \rceil}^{\infty}\frac{2^{k-1}}{4^{(p+\lambda) k}}\sum_{N=0}^\infty \frac{N^{p-1}}{4^N}\leq \frac{2^{\lceil m(1-\epsilon)\rceil-1}}{4^{(p+\lambda)\lceil m(1-\epsilon)\rceil}} \sum_{k=0}^\infty \frac{1}{2^k} \sum_{N=0}^\infty \frac{N^{p-1}}{4^N}\leq \frac{2^{ m(1-\epsilon)}}{4^{(p+\lambda)  m(1-\epsilon)}} \sum_{N=0}^\infty N^{p-1}(2^{-\frac{1}{p+\lambda}})^N. $$
Therefore,
\begin{multline*}
\sum_{L\in \lset,|L|> p}\one\biggl\{\frac{\Vert L\Vert_{1,p}+\lambda L(p+1)}{p+\lambda}>(1-\epsilon)m\biggr\}  4^{-\Vert L\Vert_{1,p}-\lambda L(p+1)}\\\leq \frac{4^{p+\lambda}2^{ m(1-\epsilon)}}{p!(p-1)!4^{(p+\lambda)  m(1-\epsilon)}} \sum_{N=p}^\infty N^{p-1}\bigl(2^{-\frac{1}{p+\lambda}}\bigr)^N.
\end{multline*}

Using the bounds on $|B_L|$ from Lemma~\ref{lem:Cperror}, we conclude that
\begin{multline}\label{eqn:lengthybound}
  \e\bigl(\var(\hat{\mu}_{\infty}-\mu\giv M)\mid H\bigr)\leq \bigl(4V_\lambda(f^{(p)})+8A\bigr)^2\frac{6}{\Pr(H)2^m} 2^{-2^{\frac{(1-\epsilon)m}{p}}} \\
  +( 4 V_\lambda(f^{(p)}))^2 \frac{4^{p+\lambda}2^{ m(1-\epsilon)}}{p!(p-1)!\Pr(H)2^m}\biggl(\,\sum_{N=p}^\infty N^{p-1}(2^{-\frac{1}{p+\lambda}})^N\biggr) 4^{-(p+\lambda) (1-\epsilon) m}.
\end{multline}

The conclusion follows
using equation (\ref{eqn:Chebshev}) from the proof of Theorem~\ref{thm:convergencerate} with the event $H$ from this proof for which $$\Pr(H^c)<(C_{p,\lambda}+e-1)2^{-\epsilon m}$$
and choosing the constant
\begin{multline*}
c=\frac{\sqrt{6}(4V_\lambda(f^{(p)})+8A)}{\sqrt{\eta}}2^{-2^{\frac{(1-\epsilon)m}{p}-1}}\\+\frac{2^{p+\lambda+2} V_\lambda(f^{(p)})}{\sqrt{p}(p-1)!\sqrt{\eta}} \biggl(\,\sum_{N=p}^\infty N^{p-1}\Bigl(\frac{1}{2}\Bigr)^{\frac{N}{p+\lambda}}\biggr)^{\frac{1}{2}} 2^{-(p+\lambda) (1-\epsilon) m}.\end{multline*}
where we used the inequality $\sqrt{a+b}\leq \sqrt{a}+\sqrt{b}$ and took square roots of the two long expressions in equation~\eqref{eqn:lengthybound} separately to make $c$ look simpler.
%\end{appendices}
\end{document}